\documentclass[11pt]{article}%
\usepackage{fullpage}
\usepackage{flushend}
\usepackage{multirow}
\usepackage{color}
\definecolor{gray}{rgb}{0.5,0.5,0.5}

\usepackage{amssymb,amsmath}
\usepackage{graphicx, epsfig}

\newcommand{\eps}{\varepsilon}

\def\ceil#1{\lceil #1 \rceil}

\newcommand{\no}[1]{}
\newcommand{\arxiv}[1]{#1}
\newcommand{\conf}[1]{}

\newcommand{\todo}[1]{} 

\newtheorem{definition}{Definition}
\newtheorem{theorem}{Theorem}
\newtheorem{lemma}{Lemma}

\newenvironment{proof}{\trivlist\item[]\emph{Proof}:}%
{\unskip\nobreak\hskip 1em plus 1fil\nobreak$\Box$
\parfillskip=0pt%
\endtrivlist}

\newcommand{\cS}{{\cal S}}
\newcommand{\cT}{{\cal T}}
\newcommand{\cTA}{{\cal TA}}
\newcommand{\cA}{{\cal A}}
\newcommand{\cX}{{\cal X}}
\newcommand{\cY}{{\cal Y}}
\newcommand{\cV}{{\cal V}}
\newcommand{\cQ}{{\cal Q}}

\newcommand{\Suf}{\mathit{Suf}}
\newcommand{\sort}{\mathit{Sort}}

\newcommand{\oT}{{\overline T}}

\newcommand{\num}{\mathrm{num}}
\newcommand{\pocc}{\mathrm{pocc}}
\newcommand{\occ}{\mathrm{occ}}
\newcommand{\rank}{\mathrm{rank}}
\newcommand{\shrank}{\mathrm{shrank}}
\newcommand{\nmargin}{\mathrm{nmargin}}
\newcommand{\pmargin}{\mathrm{pmargin}}

\newcommand{\pred}{\mathrm{pred}}
\newcommand{\bits}{\mathrm{bits}}
\newcommand{\pos}{\mathrm{pos}}
\newcommand{\bitval}{\mathrm{bitval}}
\newcommand{\newsucc}{\mathrm{succ}}
\newcommand{\up}{\mathit{UP}}

\newcommand{\cI}{{\cal I}}

\newcommand{\bT}{\mathbb{T}}

\newcommand{\longver}[1]{}

\begin{document}

\title{Text Indexing and Searching in Sublinear Time}
\author{
 J. Ian Munro\thanks{Cheriton School of Computer Science, University of Waterloo. Email {\tt imunro@uwaterloo.ca, yakov.nekrich@googlemail.com}.}
 \and
 Gonzalo Navarro\thanks{IMFD --- Millennium Institute for Foundational Research
on Data, Department of Computer Science, University of Chile. Email {\tt
gnavarro@dcc.uchile.cl}. Funded by IMFD, Mideplan, Chile.}
\and 
 Yakov Nekrich$^*$
 }
\date{}

\maketitle

\begin{abstract}
We introduce the first index that can be built in $o(n)$ time for a text of
length $n$, and can also be queried in $o(q)$ time for a pattern of length $q$. 
On an alphabet of size $\sigma$, our index uses $O(n\sqrt{\log n\log\sigma})$
bits, is built in $O(n((\log\log n)^2+\sqrt{\log\sigma})/\sqrt{\log_\sigma n})$
deterministic time, and computes the number $\occ$ of occurrences of the 
pattern in time $O(q/\log_\sigma n+\log n)$.
Each such occurrence can then be found in $O(\sqrt{\log n\log\sigma})$ time.
By slightly increasing the space and construction time, to
$O(n(\sqrt{\log n\log\sigma}+ \log\sigma\log^\eps n))$ and 
$O(n\log^{3/2}\sigma/\log^{1/2-\eps} n)$, respectively, for any constant
$0<\eps<1/2$, we can find the $\occ$ pattern occurrences in time
$O(q/\log_\sigma n + \sqrt{\log_\sigma n}\log\log n + \occ)$.
We build on a novel text sampling based on difference covers, which enjoys 
properties that allow us efficiently computing longest common prefixes in
constant time. 
\arxiv{We extend our results to the secondary memory model as well,
where we give the first construction in $o(\sort(n))$ I/Os of a data structure
with suffix array functionality; this data structure supports pattern matching queries with optimal or nearly-optimal cost.} 
\end{abstract}

\thispagestyle{empty}


\newpage
\setcounter{page}{1}
\section{Introduction}
\label{sec:intro}

We address the problem of indexing a text $T[0..n-1]$, over alphabet 
$[0..\sigma-1]$, in {\em sublinear} time on a RAM machine of $w=\Theta(\log n)$ 
bits. This is not possible when we build a classical index (e.g., a suffix
tree \cite{Weiner73} or a suffix array \cite{MM93}) that requires 
$\Theta(n\log n)$ bits, since just writing the output takes time 
$\Theta(n)$. It is also impossible when $\log\sigma = \Theta(\log n)$ and thus 
just reading the $n\log\sigma$ bits of the input text takes time $\Theta(n)$.
On smaller alphabets (which arise frequently in practice, for example on DNA,
protein, and letter sequences), the sublinear-time indexing becomes possible
when the text comes packed in words of $\log_\sigma n$ characters and we
build a {\em compressed} index that uses $o(n\log n)$ bits.
For example, there exist various indexes that use $O(n\log\sigma)$ bits 
\cite{NM06} (which is asymptotically the best worst-case size we can expect 
for an index on $T$) and could be built, in principle, in time 
$O(n/\log_\sigma n)$. Still, only linear-time indexing in compressed space 
had been achieved \cite{BCGNNalgor13,BelazzouguiN15,MNN17,MNNisaac17} until 
the very recent result of Kempa and Kociumaka \cite{KK19}.

When the alphabet is small, one may also aim at RAM-optimal pattern search,
that is, count the number of occurrences of a (packed) string $Q[0..q-1]$ in $T$
in time $O(q/\log_\sigma n)$. There exist some classical indexes using 
$O(n\log n)$ bits and counting in time $O(q/\log_\sigma n+\mathrm{polylog}(n))$
\cite{NN17,BGS17}, as well as compressed ones \cite{MNNisaac17extended}.

In this paper we introduce the first index that can be {\em built and queried
in sublinear time}. Our index, as explained, is compressed. It uses 
$O(n\sqrt{\log n \log\sigma})$ bits, which is a geometric mean between
the classical and the best compressed size. It can be built in 
$O(n((\log\log n)^2+\sqrt{\log\sigma})/\sqrt{\log_\sigma n})$ deterministic 
time. This is $o(n)$ 
for small alphabets, $\log\sigma = o(\sqrt{\log n})$, which includes 
the important practical case
where $\sigma=O(\textrm{polylog}\,n)$. Our index also supports counting queries
in $o(q)$ time. To be precise, it counts in optimal time plus an additive 
logarithmic penalty, $O(q/\log_\sigma n + \log n)$. After counting the 
occurrences of $Q$, any such occurrence can be reported in 
$O(\sqrt{\log n\log\sigma})$ time. 

A slightly larger and slower-to-build variant of our index uses
$O(n(\sqrt{\log n \log\sigma}+\log\sigma\log^\eps n))$ bits for an
arbitrarily small constant $0<\eps<1/2$ and is built in time
$O(n\log^{3/2}\sigma/\log^{1/2-\eps} n)$. This index can report the $\occ$
pattern occurrences in time
$O(q/\log_\sigma n + \sqrt{\log_\sigma n}\log\log n + \occ)$.

\begin{table}[t]
\begin{center}
\small
\begin{tabular}{l|c|c|c}
Source & Construction time & Space (bits) & Query time (counting) \\
\hline
\textcolor{gray}{Classical \cite{Wei73,McC76,Ukk95,FFM00}}
	& \textcolor{gray}{$O(n)$}
	& \textcolor{gray}{$O(n\log n)$}
	& \textcolor{gray}{$O(q\log\sigma)$} \\
\textcolor{gray}{Cole et al.~\cite{CKL15}}
	& \textcolor{gray}{$O(n)$}
	& \textcolor{gray}{$O(n\log n)$}
	& \textcolor{gray}{$O(q+\log\sigma)$} \\
\textcolor{gray}{Fischer \& Gawrychowski \cite{FG13}}$\!\!$
	& \textcolor{gray}{$O(n)$}
	& \textcolor{gray}{$O(n\log n)$}
	& \textcolor{gray}{$O(q+\log\log\sigma)$} \\
Bille et al.~\cite{BGS17} 
   	& $O(n)$
	& $O(n\log n)$ 
	& $\!\!O(q/\log_\sigma n + \log q + \log\log\sigma)$ \\
\hline
\textcolor{gray}{Classical $+$ perfect hashing}
	& \textcolor{gray}{$O(n)$ randomized}
	& \textcolor{gray}{$O(n\log n)$}
	& \textcolor{gray}{$O(q)$} \\
Navarro \& Nekrich \cite{NN17}
	& $O(n)$ randomized 
	& $O(n\log n)$
	& $O(q/\log_\sigma n + \log_\sigma^\eps n)$ \\
\hline
\textcolor{gray}{Barbay et al.~\cite{BCGNNalgor13}}
	& \textcolor{gray}{$O(n)$}
	& \textcolor{gray}{$O(n\log\sigma)$}
	& \textcolor{gray}{$O(q\log\log\sigma)$} \\
\textcolor{gray}{Belazzougui \& Navarro \cite{BelazzouguiN15}}
	& \textcolor{gray}{$O(n)$} 
	& \textcolor{gray}{$O(n\log\sigma)$}
	& \textcolor{gray}{$O(q(1+\log_w \sigma))$} \\
\textcolor{gray}{Munro et al.~\cite{MNN17,MNN17extended}}
	& \textcolor{gray}{$O(n)$} 
	& \textcolor{gray}{$O(n\log\sigma)$}
	& \textcolor{gray}{$O(q+\log\log\sigma)$} \\
Munro et al.~\cite{MNNisaac17}
	& $O(n)$ 
	& $O(n\log\sigma)$
	& $O(q+\log\log_w \sigma)$ \\
Munro et al.~\cite{MNNisaac17extended}
	& $O(n)$ 
	& $O(n\log\sigma)$
	& $\!\!O(q/\log_\sigma n+ \log n(\log\log n)^2)$ \\
\hline
\textcolor{gray}{Belazzougui \& Navarro \cite{BelazzouguiN15}}
	& \textcolor{gray}{$O(n)$ randomized}
	& \textcolor{gray}{$O(n\log\sigma)$}
	& \textcolor{gray}{$O(q(1+\log\log_w \sigma))$} \\
Belazzougui \& Navarro \cite{BelazzouguiN14}
	& $O(n)$ randomized 
	& $O(n\log\sigma)$
	& $O(q)$ \\
\hline
Kempa and Kociumaka \cite{KK19}
	& $O(n\log\sigma/\sqrt{\log n})$
	& $O(n\log\sigma)$
	& $O(q(1+\log_w\sigma))$ \\
\hline
\textbf{Ours}
	& $\!\!O\left(n\,\frac{(\log\log
n)^2+\sqrt{\log\sigma}}{\sqrt{\log_\sigma n}}\right)\!\!$
	& $\!\!O(n\sqrt{\log n \log\sigma})\!\!$
	& $O(q/\log_\sigma n + \log n)$ \\
\end{tabular}
\end{center}
\caption{Previous and our results for index construction on a text
of length $n$ and a search pattern of length $q$, over an alphabet of size
$\sigma$, on a RAM machine of $w$ bits, for any constant $\eps>0$. Grayed 
rows are superseded by a more recent result in all aspects we consider. Note
that $O(n)$-time randomized construction can be replaced by $O(n(\log\log n)^2)$
deterministic constructions \cite{Ruz08}.}
\label{tab:compar}
\end{table}

As a comparison (see Table~\ref{tab:compar}), the other indexes that count
in time $O(q/\log_\sigma n + \textrm{polylog}(n))$ use either more space 
($O(n\log n)$ bits) and/or construction time ($O(n)$)
\cite{BGS17,NN17,MNNisaac17extended}.
The indexes using less space, on the other hand, use as little as
$O(n\log \sigma)$ bits but are slower to build and/or to query
\cite{MNN17,MNN17extended,MNNisaac17,MNNisaac17extended,BelazzouguiN14,KK19}.
A recent construction \cite{KK19} is the only one able
to build in sublinear time ($O(n\log\sigma/\sqrt{\log n})$ time)
and to use compressed space ($O(n\log\sigma)$ bits, less than ours), but
it is still unable to search in $o(q)$ time.
Those compressed indexes can then deliver each occurrence in
$O(\log^\eps n)$ time, or even in $O(1)$ time if a structure of 
$O(n\log^{1-\eps}\sigma\log^\eps n)$ further bits is added, though there is
no sublinear-time construction for those extra structures either
\cite{Rao02,GNF14}.

Our technique is reminiscent to the Geometric BWT \cite{Chien2015}, where a
text is sampled regularly, so that the sampled positions can be indexed with
a suffix tree in sublinear space. In exchange, all the possible alignments of the 
pattern and the samples have to be checked in a two-dimensional range search 
data structure. To speed up the search, we use instead an irregular sampling 
that is dictated by a {\em difference cover}, which guarantees that for any two
positions there is a small value that makes them sampled if we shift both by
the value. This enables us to compute in constant time the longest common prefix
of any two text positions with only the sampled suffix tree. With this 
information we can efficiently find the locus of each alignment from the 
previous one.
Difference covers were used in the past for suffix tree construction 
\cite{KSB06,GK17}, but never for improving query times.


\section{Preliminaries and Difference Covers}
\label{sec:diffcov}

We denote by $|S|$ the number of symbols in a sequence $S$ or the number of elements in a set $S$. For two strings $X$ and $Y$, $LCP(X,Y)$ denotes the longest common prefix of $X$ and $Y$. For a string $X$ and a set of strings $\cS$, $LCP(X,\cS)=\max_{Y \in \cS} LCP(X,Y)$, where we compare lengths to take the maximum. We assume that the concepts associated with
suffix trees \cite{Weiner73} are known. We describe in more detail the
concept of a difference cover.

\begin{definition}
Let $D=\{\,a_0,a_1,\ldots, a_{t-1}\,\}$ be a set of $t$ integers satisfying $0\le a_i\le s$. We say that $D$ is a \emph{difference cover modulo $s$} of size $t$, denoted $DC(s,t)$, if for every  $d$ satisfying $1\le d\le s-1$ there is a pair $0\le i,j\le t-1$ such that $d=(a_i-a_j)\!\!\mod s$. 
\end{definition}

Colburn and Ling~\cite{ColbournL00} describe a difference cover $DC(\Theta(r^2),\Theta(r))$ for any positive integer $r$ that is based on a result of Wichmann~\cite{Wichmann63}. Consider a sequence $b_1,\ldots, b_{6r+3}$ where 
$b_i=1$ for $1\le i\le r$, $b_{r+1}=r+1$,  $b_i=2r+1$ for $r+2\le i \le 2r+1$, 
$b_i=4r+3$ for $2r+2\le i \le 4r+2$, $b_i=2r+2$ for $4r+3\le i \le 5r+ 3$, 
and $b_i=1$ for $5r+4\le i \le 6r+3$. We set  $a_0=0$ and $a_i=a_{i-1}+b_i$ for  $i=1,\ldots, 6r+3$.
\begin{lemma}\cite{ColbournL00}
\label{lemma:dc1}
  The set $A=\{\,a_0,\ldots,a_{6r+3}\,\}$ is a difference cover $DC(12r^2+ 18r+ 6, 6r+4)$.
\end{lemma}
The following well-known property of difference covers will be extensively used in our data structure. 
\begin{lemma}
\label{lemma:dc2}
  Let $D=\{\,a_0,\ldots,a_{t-1}\,\}$  be a difference cover modulo $s$. Then 
for any $1\le i,j \le s$  there exists a non-negative integer  $h(i,j)< s$ such that 
$(i+h(i,j)) \!\!\mod s \in D$ and $(j+h(i,j)) \!\!\mod s \in D$. If $D$ is known, we can compute $h(i,j)$ for all $i,\, j$ in time $O(s^2)$.  
\end{lemma}
\begin{proof}
  For any $x$ there exists a pair $f_x\in D$, $e_x\in D$ satisfying $e_x-f_x=x\!\!\mod s$ by definition of a difference cover. Let $d[x]=f_x$. Then both $d[x]$ and $(d[x]+x)\!\!\mod s$ are in $D$. 
Let $h(i,j)=(d[(j-i)\!\!\mod s]-i)\!\!\mod s$. Then $i+h(i,j) \!\!\mod s=d[(j-i)\!\!\mod s]\in D$ and $j+h(i,j) \!\!\mod s =d[(j-i)\!\!\mod s]+ (j-i)\!\!\mod s\in D$. 
\end{proof}

\section{Data Structure}
\label{sec:datastruc}

We divide the text $T[0..n-1]$ into blocks of $\Theta(\log_{\sigma} n)$ consecutive symbols. To be precise, every block consists of $s=12r^2+18r +6$ symbols where  $r=\Theta(\sqrt{\log_{\sigma }n})$ for an appropriately chosen small constant. 
To avoid tedious details, we assume that $\sqrt{\log_{\sigma} n}$ is an integer and that the text length is divisible by $s$.  We select $6r+4$ positions in every block that correspond to 
the difference cover of Lemma~\ref{lemma:dc1}. That is, all positions $si+a_j$ for $i=0,1,\ldots,(n/s)-1$ and $0\le j\le 6r+3$ are selected. The total number of selected positions is $O(n/r)$. 
The set $\cS'$ consists of all suffixes starting at selected positions. 
A substring  between two consecutive selected symbols $T[si+a_j..si+a_{j+1}-1]$ is called a \emph{sub-block}; note that a sub-block may be as long as $4r+3$. Our data structure consists of the following three components. 
\begin{enumerate}
\item 
  The suffix tree $\cT'$ for suffixes starting at selected positions, which uses $O((n/r)\log n)=O(n\sqrt{\log n\log\sigma})$ bits. Thus $\cT'$ is a compacted trie for the suffixes in  $\cS'$. Suffixes are represented as strings of meta-symbols where every meta-symbol corresponds to a substring of $\log_{\sigma} n$ consecutive symbols. Deterministic dictionaries are used at the nodes to descend 
by the meta-symbols in constant time. Predecessor structures are also used
at the nodes, to descend when less than a metasymbol of the pattern is left.
Given a pattern $Q$, we can identify all selected suffixes starting with $Q$ in $O(|Q|/\log_{\sigma} n)$ time, plus an $O(\log\log n)$ additive term coming
from the predecessor operations at the deepest node.
\item
  A data structure on a set $\cQ$ of points. Each point of $\cQ$ corresponds to a pair $(ind_i, rev_i)$ for $i=1,\ldots, (n/s)-1$ where $ind_i$ is the index of the $i$-th selected suffix of $T$ in the lexicographically sorted set $\cS'$ and $rev_i$ is an integer that corresponds to the reverse sub-block preceding that $i$-th selected suffix in $T$.  Our data structure supports two-dimensional range counting and reporting queries on $\cQ$, exploiting the fact that the $y$-coordinates are in the range $[0..\sigma^{4r+3}]$.
\item
	A data structure for \emph{suffix jump} queries on $\cT'$. Given a string $Q[0..q-1]$, its locus
node $u$, and a positive integer $i\le 4r+3$, a (suffix) $i$-jump query returns the locus of $Q[i..q-1]$, or it says that $Q[i..q-1]$ does not prefix any string in $\cS'$.
 The suffix jump structure has essentially the same functionality as the suffix links, but we do not store suffix links explicitly in order to save space and improve the construction time.
\end{enumerate}

As described, $\cT'$ is a trie over an alphabet of meta-symbols corresponding to strings of length $\log_{\sigma}n$. Therefore, a string $Q$ does not necessarily have a single locus node. We will then denote its locus as $u[l..r]$, where $u \in \cT'$ is the deepest node whose string label is a prefix of $Q$ and $[l..r]$ is the maximal interval such that the string labels of the children $u_l,\ldots,u_r$ of $u$ are prefixed by $Q$. 

Using our structure, we can find all the occurrences in $T$ of a pattern $Q[0..q-1]$ whenever $q> 4r+3$. Occurrences of $Q$ are classified according to their positions relative to selected symbols. An occurrence $T[f..f+q-1]$ of $Q$ is an $i$-occurrence if $T[f+i]$ (corresponding to the $i$-th symbol of $Q$) is the leftmost selected symbol in $T[f..f+q-1]$.

First, we identify all $0$-occurrences by looking for $Q$ in $\cT'$:  We
traverse the path corresponding to $Q$ in $\cT'$ to find $Q_0=LCP(Q,\cS')$,
the longest prefix of $Q$ that is in $\cT'$, with locus $u_0[l_0..r_0]$.
Let $q_0 =|Q_0|$; if $q_0=q$, then $u_0[l_0..r_0]$ is the locus of $Q$ and we
count or report all its $0$-occurrences as the positions of suffixes in the subtrees of $u_0[l_0..r_0]$.\footnote{For fast counting, each node may also store the cumulative sum of its preceding siblings.} If $q_0<q$, there are no $0$-occurrences of $Q$. 

Next, we compute a $1$-jump from $u_0$ to find the locus of
$Q_0[1..]=Q[1..q_0-1]$ in $\cT'$. If the locus does not exist, then there are
no 1-occurrences of $Q$. If it exists, we traverse the path in $\cT'$ starting
from that locus. Let $Q_1=Q[1..q_1-1] = LCP(Q[1..q-1],\cS')$ be the longest prefix of $Q[1..q-1]$ found in $\cT'$, with locus
$u_1[l_1..r_1]$. If $q_1<q-1$, then again there are no 1-occurrences of $Q$. If $q_1=q-1$, then $u_1[l_1..r_1]$ is the locus of $Q[1..q-1]$. In this case, every $1$-occurrence of $Q$ corresponds to an occurrence of $Q_1$ in $T$ that is preceded by
$Q[0]$. We can identify them by answering a two-dimensional range query $ [ind_1,ind_2]\times [rev_1,rev_2]$ where $ind_1$ ($ind_2$) is the leftmost (rightmost) leaf in the subtrees of $u_1[l_1..r_1]$ and $rev_1$ ($rev_2$) is the smallest (largest) integer value of any reverse sub-block that starts with $Q[0]$. 

In order to avoid technical complications we maintain four data structures for range reporting queries. Points representing 
(suffix, sub-block) pairs are classified according to the sub-block size.  If
a selected suffix $T[f..]$ is preceded by another selected suffix $T[f-1]$,
then we do not need to store a point for this suffix. In all other cases the
size of the sub-block that precedes a suffix is either $r$ or $2r$ or $4r+2$
or $2r+1$. We assign the point representing the (sub-block,suffix) pair to one
of the four data structures. When we want to count or report $i$-occurrences of $Q$, we
query up to four data structures, that is, all those storing sub-blocks
of length $\ge i$.

We proceed and consider $i$-occurrences for $i=2,\ldots, 4r+3$ using the same
method. Suppose that we have already considered the possible $j$-occurrences
of $Q$ for $j=0,\ldots,i-1$. Let $t$ be such that $q_t=\max(q_0,\ldots,
q_{i-1})$, considering only the loci that exist in $\cT'$. We compute the
$(i-t)$-jump from $u_t$, where $u_t[l_t..r_t]$ is the locus of
$Q_t=Q[t..q_t-1]=LCP(Q[t..q-1],\cS')$. If $Q[i..q_t-1]$ is found in $\cT'$,
with locus $u[l..r]$, but $q_t<q-1$, we traverse from $u$ downwards to
complete the path for $Q[i..q-1]$. We then find the locus $u_i[l_i..r_i]$ of $Q[i..q_i-1]=LCP(Q[i..q-1],\cS')$. If $q_i=q$, then $Q[i..q-1]$ is found, so we count or report all $i$-occurrences by answering a two-dimensional query as described above.  

\paragraph{Analysis.}
The total time is $O(q/\log_{\sigma}n+r(1+t_q+t_s))$, where $t_q$ and $t_s$ are
the times to answer a range query and to compute a suffix jump, respectively.
All the downward steps in the suffix tree amortize to $O(q/\log_\sigma n + r)$:
we advance $q_t$ by $\log_\sigma n$ in each downward step, but before
taking each suffix jump we can reduce $q_t$ by less than $\log_\sigma n$ because
we take it from $u_t$, whereas the actual locus with string depth $q_t$ is 
$u_t[l_t..r_t]$. As said, the suffix tree (point 1) uses 
$O(n\sqrt{\log n\log\sigma})$ bits.

The data structure of point 2 is a wavelet tree \cite{Chazelle88,GGV03,Navjda13}
built on $t=O(n/r)$ points. Its height is the logarithm of the $y$-coordinate
range, $h=\log(\sigma^{4r+3}) = O(r\log\sigma) = O(\sqrt{\log n\log\sigma})$,
and it uses $O(t\cdot h) = O(n\log\sigma)$ bits. 
Such structure answers range counting queries in time $t_q = O(h) = 
O(\sqrt{\log n\log\sigma})$, thus $r \cdot t_q = O(\log n)$, and reports each
point in the range in time $O(h) = O(\sqrt{\log n\log\sigma})$. 

In Sections~\ref{sec:jumps} and \ref{sec:shortjump} we show how to
implement all the $r$ suffix jumps (point 3) in time $r \cdot t_s = 
O(q/\log_\sigma n + r\log\log n)$ time, with a structure that uses
 $O(n\sqrt{\log n\log\sigma})$ bits of space.

Section~\ref{sec:construction} shows that the construction time
of the structures of point 1 is $O(n(\log\log n)^2/\sqrt{\log_\sigma n})$ and
of point 3 is $O(n/\sqrt{\log_\sigma n})$. The wavelet tree of point 2
can be built in time 
$O(t \cdot h / \sqrt{\log t}) = O(n\log\sigma/\sqrt{\log n})$ \cite{MNV16,BabenkoGKS15}. 

Finally, since a pattern shorter than $4r+4$ may not cross a sub-block boundary and thus we could miss some occurrences, Section~\ref{sec:small} describes a special index for those small patterns. Its space and construction
time is also within those of point 3.  This yields our first result.

\begin{theorem}
  \label{theor:ram0}
Given a text $T$ of length $n$ over an alphabet of size $\sigma$, we can build
an index using $O(n\sqrt{\log n \log\sigma})$ bits in deterministic
time $O(n((\log\log n)^2+\sqrt{\log\sigma}) / \sqrt{\log_\sigma n})$, so that it can count the number of occurrences of a pattern of length $q$ in time 
$O(q/\log_\sigma n + \log n)$, and report each such occurrence in time $O(\sqrt{\log n \log\sigma})$.
\end{theorem}

We can improve the time of reporting occurrences by slightly increasing the construction time. 
Appendix~\ref{sec:rangerep} shows how to construct a 
range reporting data structure (point 2) that, after $t_q=O(\log\log n)$ time, 
can report each occurrence in constant time. The space of this structure is 
$O(n\log\sigma\log^\eps n)$ bits and its construction time is 
$O(n\log^{3/2}\sigma / \log^{1/2-\eps} n)$, for any constant $0<\eps<1/2$. This yields our second result, which is of interest for reporting the pattern occurrences.

\begin{theorem}
  \label{theor:ram1}
Given a text $T$ of length $n$ over an alphabet of size $\sigma$, we can build
an index using $O(n(\sqrt{\log n \log\sigma}+\log\sigma\log^\eps n))$ bits in 
time $O(n\log^{3/2}\sigma / \log^{1/2-\eps} n)$, for any constant $0<\eps<1/2$,
so that it can report the $\occ$ occurrences of a pattern of length $q$ in time 
$O(q/\log_\sigma n + \sqrt{\log_\sigma n}\log\log n + \occ)$.
\end{theorem}

\section{Suffix Jumps}
\label{sec:jumps}

Now we show how suffix jumps can be implemented. The solution described in
this section takes $O(\log n)$ time per jump and is used when $|Q|\ge \log^3 n$.
This already provides us with an optimal solution because, in this case, the 
time of the $r$ suffix jumps, $O(\log n \sqrt{\log_\sigma n})$, is
subsumed by the time $O(q/\log_\sigma n)$ to traverse the pattern. 
In the next section we describe an appropriate method for short patterns.

Given a substring $Q_t[0..q_t-1]$ of the original query $Q$, with known locus
node $u_t$, we need to find the locus of $Q_t[i..]$ or determine that it does 
not exist. 
Note that our answer is of the form $v[l..r]$. Once the node $v$ is known, we can compute $[l..r]$ in $O(\log\log n)$ time with predecessor queries, which will not impact in the time to compute a suffix jump. Therefore, we focus in how to
compute $v$.
Our solution is based on using the properties of difference covers.
We start by showing how to efficiently compute certain $LCP$ values.


\begin{lemma}
  \label{lemma:lcp}
Let $S_1=T[f_1..]$ and $S_2=T[f_2..]$ be two suffixes from the set $\cS'$.
We can compute $|LCP(T[f_1+i..],S_2)|$ for any $1\le i \le 4r+3$ in $O(1)$ time.
\end{lemma}
\begin{proof}
  Let $\delta=h(0,i)$ where the function $h(j,i)$ is as defined in Lemma~\ref{lemma:dc2}. By definition, both $T[f_1+i+\delta..]$ and $T[f_2+\delta..]$ are in $\cS'$. We can find their corresponding leaf in $\cT'$ by storing a pointer to
$\cT'$ for each sampled text position, which requires other
$O(n\sqrt{\log n\log\sigma})$ bits. Hence we can compute 
$\ell=|LCP(T[f_1+i+\delta..], T[f_2+\delta..])|$ in $O(1)$ time with a lowest
common ancestor query \cite{BFCPSS05}. 
Since $\delta\le s = O(\log_\sigma n)$,  the
substrings $T[f_1+i..f_1+i+\delta-1]$ and $T[f_2..f_2+\delta-1]$ fit into
$O(1)$ words. Hence, we can compute the length $\ell'$ of their  longest common prefix in $O(1)$ time. If $\ell'< \delta$, then $|LCP(T[f_1+i..],S_2)|=\ell'$. Otherwise 
$|LCP(T[f_1+i..],S_2)|=\delta+\ell$. 
\end{proof}

We compute the locus  of $Q_t[i..]$  by applying Lemma~\ref{lemma:lcp} $O(\log n)$
times; note that we know the text position $f_1$ of an occurrence of $Q_t$
because we know its locus $u_t$ in $\cT'$; therefore $Q_t[i..]=T[f_1+i..]$. 
By binary search among the sampled suffixes (i.e., leaves of $\cT'$),
we identify in $O(\log n)$ time the suffix $S_m$ that maximizes
$|LCP(Q_t[i..],S_m)|$, because this measure decreases monotonically in both
directions from $S_m$. At each step of the binary search we compute
$l=|LCP(Q_t[i..],S)|$ for some suffix $S\in\cS'$ using Lemma~\ref{lemma:lcp} and compare their $(l+1)$th symbols to decide the direction of the binary search. 
Next we find, again with binary search, the smallest and largest suffixes
$S_1,S_2 \in \cS'$ such that $|LCP(S_1,S_m)|=|LCP(S_2,S_m)|=|LCP(Q_t[i..],S_m)|$;
note $S_1 \le S_m \le S_2$.  

Finally let $v$ be the lowest common  ancestor of the leaves that hold $S_1$
and $S_2$ in $\cT'$. It then
holds that $LCP(Q_t[i..],\cS')=LCP(Q_t[i..],S_m)$, and $v$ is its locus node. Further,
if $|LCP(S_m,Q_t[i..])|=q_t-i=|Q_t[i..]|$, then $v$ is also the locus of $Q_t[i..]$; otherwise $Q_t[i..]$ prefixes no string in $\cS'$.

\begin{lemma}
  \label{lemma:longjump}
Suppose that we know $Q_t[0..q_t-1]$ and its locus in $\cT'$. We can then 
compute $LCP(Q_t[i..q_t-1],\cS')$ and its locus in $\cT'$ in $O(\log n)$ time, for
any $0\le i\le 4r+3$.
\end{lemma}

%

\section{Suffix Jumps for Short Patterns}
\label{sec:shortjump}

In this section we show how $r$ suffix jumps can be computed in 
$O(|Q|/\log_\sigma n + r \log \log n)$ time when $|Q|\le \log^3 n$. 
Our basic idea is to construct a set $\cX_0$ of selected substrings with
length up to $\log^3 n$. These are sampled at polylogarithmic-sized intervals
from the sorted set $\cS'$. We also create a superset $\cX\supset \cX_0$ that
contains all the substrings that could be obtained by trimming the first $i\le 4r+3$ symbols from strings in $\cX_0$.
Using lexicographic naming and special dictionaries on $\cX$, we pre-compute
answers to all suffix jump queries for strings from $\cX_0$. We start by
reading the query string $Q$ and  trying to match $Q$, $Q[1..]$, $Q[2..]$ in
$\cX_0$. That is, for every $Q[i..q-1]$ we find $LCP(Q[i..q-1],\cX_0)$ and its
locus in $\cT'$. With this information we can finish the computation of a suffix jump in 
$O(\log\log n)$ time, because the information on $LCP$s in $\cX_0$ will narrow
down the search in $\cT'$ to a polylogarithmic sized interval, on which we can
use the binary search of Section~\ref{sec:jumps}.  


\paragraph{Data Structure.}
Let $\cS''$ be the set obtained by sorting suffixes in $\cS'$ and selecting  every $\log^{10} n$-th suffix. We denote by $\cX$ the set of  all substrings $T[i+f_1..i+f_2]$ such that the suffix $T[i..]$ is in the set $\cS''$ and $0\le f_1\le f_2\le \log^3 n$. We denote by $\cX_0$ the set of substrings $T[i..i+f]$ such that the suffix $T[i..]$ is in the set $\cS''$ and $0\le f\le \log^3 n$. 
Thus $\cX_0$ contains all prefixes  of length up to $\log^3 n$ for  all suffixes from $\cS''$ and $\cX$ contains all strings that could be obtained by suffix jumps from strings in $\cX_0$. 


We assign unique integer names to all substrings in $\cX$: we sort $\cX$ and then traverse the sorted list assigning a unique integer $\num(S)$ to each substring $S \in \cX$. Our goal is to store pre-computed solutions to suffix jump queries. To this end, we keep three dictionary data structures: 
\begin{itemize}
\item 
Dictionary $D_0$ contains the names $\num(S)$ for all $S\in \cX_0$, as well as their loci in $\cT'$.  
\item 
Dictionary $D$ contains the names $\num(S)$ for all substrings $S\in \cX$. For every entry $x\in D$, with $x=\num(S)$, we store (1) the length $\ell(S)$ of the string $S$, (2) the length $\ell(S')$ and the name $\num(S')$ where $S'$ is the longest prefix of $S$ satisfying $S'\in \cX_0$, 
(3) for each $j$, $1 \le j \le 4r+3$ , the name $\num(S[j..])$ of the string obtained by trimming the first $j$ leading symbols of $S$ if $S[j..]$ is in $\cX$. 
\item 
Dictionary $D_p$ contains all pairs $(x,\alpha)$, where $x$ is an integer and $\alpha$ is a string, such that the length of $\alpha$ is at most $\log_{\sigma} n$, $x=\num(S)$ for some $S\in\cX$, and the concatenation $S\cdot\alpha$ is also in $\cX$. $D_p$ can be viewed as a (non-compressed) trie on $\cX$. 
\end{itemize}

Using $D_p$, we can navigate among the strings in $\cX$: if we know $\num(S)$ for some $S\in \cX$, we can look up the concatenation $S\alpha$ in $\cX$ for any string $\alpha$ of length at most $\log_{\sigma} n$. The dictionary $D$ enables us to compute suffix jumps between strings in $\cX$: if we know $\num(S[0..])$ for some $S\in \cX$, we can look up $\num(S[i..])$ in $O(1)$ time. 

The set $\cS''$ contains $O(\frac{n}{r\log^{10}n})$ suffixes. The set $\cX$
contains $O(\log^6n)$ substrings for every suffix in $\cS''$. The space usage
of dictionary $D$ is $O(n/\log^4 n)$ words, dominated by item (3). 
The space of $D_p$ is $O(n\log_{\sigma}n/(r\log^4 n))$ words, given by the 
number of strings in $\cX$ times $\log_{\sigma} n$. This dominates the total 
space of our data structure, $O(n\sqrt{\log_{\sigma} n}/\log^4n)$. 


\paragraph{Suffix Jumps.} 
Using  the dictionary $D$, we can compute suffix jumps within $\cX_0$. 
\begin{lemma}
\label{lemma:jumpx0}
  For any string $Q$ with $4r+3<|Q|\le \log^3 n$, we can find the strings 
$P_i=LCP(Q[i..],\cX_0)$, their lengths $p_i$ and their loci in $\cT'$, for all $1\le
i\le 4r+3$, in time $O(|Q|/\log_{\sigma}n+r\log\log_{\sigma}n)$. 
\end{lemma}
\begin{proof}
  We find  $P_0=LCP(Q[0..q-1],\cX_0)$ in
$O(|P_0|/\log_{\sigma}n+\log\log_{\sigma}n)$ time: suppose that $Q[0..x]$
occurs in $\cX_0$. We can check whether $Q[0..x+\log_{\sigma}n]$ also occurs
in $\cX_0$ using the dictionary $D_p$. If this is the case, we increment $x$
by $\log_{\sigma}n$.  Otherwise we find with binary search, in
$O(\log\log_{\sigma}n)$ time,  the largest $f\le \log_{\sigma}n$ such that
$Q[0..x+f]$ occurs in $\cX_0$. Then $P_0=Q[0..x+f] \in \cX_0$, and its locus
in $\cT'$ is found in $D_0$.

When $P_0$, of length $p_0=|P_0|$, and its name $\num(P_0)$ are known, we 
find $P_1 = LCP(Q[1..],\cX_0)$: first we look up $v=\num(P_0[1..])$ 
in component (3) 
of $D$, then we look up in component (2) of $D$ the longest 
prefix of the string with name $v$ that is in $\cX_0$. This is the 1-jump of
$P_0$ in $\cX_0$; now we descend as much as possible from there using $D_p$,
as done to find $P_0$ from the root. We finally obtain $\num(P_1)$; 
its length $p_1$ and locus in $\cT'$ are found in $D$
(component (1)) and $D_0$, respectively. 

We proceed in the same way as in Section~\ref{sec:datastruc} and find
$LCP(Q[i..],\cX_0)$ for $i=2$, $\ldots$, $4r+3$. The traversals in $D_p$ 
amortize analogously to $O(|Q|/\log_\sigma n+r)$, and we have
$O(r\log\log_\sigma)$ further time to complete the $r$ traversals.
\end{proof}

When all $LCP(Q[i..],\cX_0)$ and their loci in $\cT'$ are known, we can compute suffix jumps in $\cS'$. 

\begin{lemma}
\label{lemma:jumpshort}
 Suppose that we know $P_i=LCP(Q[i..q-1],\cX_0)$ and its locus in $\cT'$ for all
$0 \le i\le 4r+3$. Assume we also know $Q_t[0..q_t-1] \in \cS'$ and its
locus node $u_t \in \cT'$, where $Q_t=Q[t..t+q_t-1]$. Then, given 
$j \le 4r+3$, we can compute $LCP(Q_t[j..],\cS')$ and its locus in $\cT'$, in 
$O(\log \log n)$ time.
\end{lemma}
\begin{proof}
Let $v[l..r]$ be the locus of $LCP(Q_t[j..],\cX_0)=LCP(Q[t+j..],\cX_0)$ in 
$\cT'$ and let
$\ell=|LCP(Q_t[j..],\cX_0)|$. If $\ell=q-j$, then $v[l..r]$ is the locus of
$Q_t[j..]$ in $\cT'$. Otherwise let $v_1$ denote the child of $v$   that is
labeled with  $Q[t+j+\ell..t+j+\ell+\log_{\sigma}n-1]$. If $v_1$ does not
exist, then $v$ is the locus of $LCP(Q_t[j..],\cS')$. If $v_1$ exists, then
the locus of $LCP(Q_t[j..],\cS')$ is in the subtree $\cT_{v_1}$ of $\cT'$ 
rooted at $v_1$.   By definition, $\cT_{v_1}$ does not contain suffixes from $\cX_0$. Hence $\cT_{v_1}$ has $O(\log^{10}n)$ leaves. 
We then find $LCP(Q_t[j..],\cS')$ among suffixes in $\cT_{v_1}$ using the
binary search method described in Section~\ref{sec:jumps}. We find $S_1$,
$S_m$, and $S_2$ in time $O(\log \log^{10}n) = O(\log\log n)$.
The locus of $LCP(Q_t[j..],\cS')$ is then the lowest common ancestor of the
leaves that hold $S_1$ and $S_2$.
%
%
\end{proof}

\begin{lemma}
  \label{lemma:shortloci}
Suppose that $|Q|\le \log^3 n$. Then we  can find all the existing loci of
$Q[i..]$ in $\cT'$, for $0\le i < 4r+3$, in time $O(|Q|/\log_{\sigma}
n+r\log\log n)$.
\end{lemma}


\section{Construction}
\label{sec:construction}

\paragraph{Sampled suffix tree.}
Consider the set of strings $P_{i,j}=T[si+a_j..s(i+1)+a_j-1]$, that is, all substrings of length $s$ that start at selected positions (any $P_{i,j}$ shifted by $a_j$). We view each $P_{i,j}$ as a string of length $s/\log_{\sigma}n$ over an alphabet of size $n$. In other words, we view $P_{i,j}$ as a sequence of meta-symbols, such that each meta-symbol represents $\log_{\sigma}n$ symbols and fits into one word. 

We assign lexicographic names to the strings and construct texts $T_j= P_{0,j} P_{1,j} \cdots$ for $j=0$, $\ldots$, $4r+3$. That is, the $i$-th symbol in $T_j$ is the lexicographic name of the string $P_{i,j} = T[a_j+is..a_j+(i+1)s-1]$. 
Let $\oT$ denote the concatenation of texts $T_j$, $\oT=T_1\$T_2\$\ldots T_{4r+3}\$ $, where \$ stands for a lexicographic name smaller than all those used.

Since $\oT$ consists of $O(n/r)$ meta-symbols and each meta-symbol fits in a $(\log n)$-bit word, we can sort all meta-symbols in $O(n/r)$ time. Thus we can generate $\oT$ and construct its suffix tree $\cT'$ in $O(n/r)$ time \cite{FFM00}. Further, we need 
$O((n/r)(\log\log n)^2) = O(n(\log\log n)^2/\sqrt{\log_\sigma n})$ time to
build the perfect hashes and predecessor data structures with the children of
each node \cite{Ruz08,BBV10}.

\paragraph{Suffix jumps.}
The lowest common ancestor structure \cite{BFCPSS05} and the pointers from
sampled text positions to leaves in $\cT'$, needed in
Section~\ref{sec:jumps}, are built in time
$O(|\cT'|) = O(n/\sqrt{\log_\sigma n})$.

The sets of substrings and dictionaries $D$, $D_0$, and $D_p$ described in Section~\ref{sec:shortjump} can be constructed  as follows. Let $m=O(n/r)$ be the number of selected suffixes in $\cS'$. The number of suffixes in $\cS''$ is $O(m/\log^{10}n)$. The number of substrings associated with each suffix in $\cS''$ is $O(\log^6n)$ and their total length is $O(\log^9n)$. The total number of strings in $\cX_0$ is $O(m/\log^{7} n)$ and their total length is $O(\frac{m}{\log^{10}n}\cdot \log^6 n)=O(m/\log^4 n)$. The number of strings in $\cX$ is $k=O((m/\log^{10} n)\cdot \log^6n)=O(m/\log^4 n)$  and their total length is $t=O((m/\log^{10} n)\cdot \log^9n)=O(m/\log n)$. 
We can sort all the strings in $\cX$ and compute their lexicographic names in $O(t)=O(m/\log n)$ time using RadixSort. 

Next, we construct the dictionary $D_0$ that contains names $\num(S)$ of all $S\in \cX_0$. For every $x=\num(S)$ in $D_0$ we keep a pointer to the string $S$ and its locus in $\cT'$. 
$D_0$ is a perfect hash table on the keys $\num(S)$, so it can be constructed in $O(|\cX_0|(\log\log n)^2)$ time~\cite{Ruz08}. 

When we generate strings of  $\cX$, we also record the information about suffix jumps, so that we have the pointers to all relevant suffix jumps for each string $S\in\cX$. Using the pointers to the strings $S$ and the dictionary $D_0$, we can determine the longest prefix $S' \in \cX_0$ of $S$ by binary search on $\ell(S')$, in $O(\log\log n)$ time. We then have all the information associated with elements of $D$ (items (1)--(3)). 
 The  dictionary $D$ itself is a perfect hash table on the $k$ keys, which can be constructed  in $O(k(\log\log n)^2)=o(m)$ time~\cite{Ruz08}. 

Finally, we construct the  dictionary $D_p$ by inserting all strings in $\cX$ into a trie data structure; for every node of this trie we store the name $\num(S)$ of the corresponding string $S$. Along a depth-first trie traversal we collect, for each node representing name $y$, its ancestors $x$ up to distance 
$\log_\sigma n$ and the strings $\alpha$ separating $x$ from $y$. All the pairs
$(x,\alpha) \rightarrow y$ are then stored in $D_p$. Since $\cX$ is 
prefix-closed, the trie contains $O(k)$ nodes, and we include
$O(k \log_\sigma n)$ pairs in $D_p$. Since $D_p$ is also a perfect hash table,
it can be built in time $O(k \log_\sigma n (\log\log n)^2)=o(m)$.

Hence the total time needed to construct the data structures for suffix jumps is $O(m)=O(n/r)=O(n/\sqrt{\log_\sigma n})$. 

\paragraph{Range searches.}
As said, the wavelet tree can be built in time $O(n\log\sigma/\sqrt{\log n})$ 
\cite{MNV16,BabenkoGKS15}. Appendix~\ref{sec:rangerep} shows that the time to
build the data structure for faster reporting is
$O(n\log^{3/2}\sigma/\log^{1/2-\eps}n)$, for any constant $0<\eps<1/2$.

\section{Index for Small Patterns}

\label{sec:small}

The data structure for small query strings consists of two tables. Let $p=4r+3$.
The value of the parameter $r=\Theta(\sqrt{\log_{\sigma }n})$ is selected so 
that $2p\le (1/2)\log_{\sigma}n$. We regard the text as an array $A[0..n/p]$ of
length-$2p$ (overlapping) strings, $A[i]=T[ip..ip+2p-1]$. A table $Tbl$ is built whose
entries correspond to all strings of length $2p$: $Tbl[\alpha]$ contains all
the positions $i$ where $A[i]=\alpha$. Further, we build tables $Tbl_j$, for
$1 \le j \le p$, containing all the possible length-$j$ strings.
Each entry $Tbl_j[\beta]$, with $|\beta|=j$, contains the list of
length-$2p$ strings $\alpha$ such that $Tbl[\alpha]$ is not empty and $\beta$
is a substring of $\alpha$ beginning in its first $p$ positions(i.e., 
$\beta=\alpha[i..i+j-1]$ for some $0\le i< p$). 

Table $Tbl$ has $\sigma^{2p}=O(\sqrt{n})$ entries, and overall contains
$n/p = O(n/r)$ pointers to $A$, thus its total space is 
$O(n\sqrt{\log n\log\sigma})$ bits.
Tables $Tbl_j$ add up to $O(\sigma^p) = O(n^{1/4})$ cells. Since each distinct
string $\alpha$ of length $2p$ produces $O(r^2)$ distinct substrings, there
can be only $O(\sigma^{2p} r^2) = O(\sqrt{n} \log_\sigma n)$ pointers in all
the tables $Tbl_j$, for a total space of $o(n)$ bits.

To report occurrences of a query string $Q[0..q-1]$, we examine the list 
$Tbl_q[Q]$. For each string $\alpha$ in $Tbl_q[Q]$, we visit the entry 
$Tbl[\alpha]$ and report all the positions of $Tbl[\alpha]$ in $A$ (with their 
offset).

To build $Tbl$, we can traverse $A$ and add each $i$ to the list of
$Tbl[A[i]]$, all in $O(n/r)$ time. We then we visit the slots of $Tbl$. For 
every $\alpha$ such that $Tbl[\alpha]$ is not empty, we consider all the
sub-strings $\beta$ of $\alpha$ and add $\alpha$ to $Tbl_{|\beta|}[\beta]$,
recording also the corresponding offset of $\beta$ in $\alpha$ (we may add
the same $\alpha$ several times with different offsets). The time of this
step is, as seen for the space, $O(\sigma^{2p} r^2) = O(\sqrt{n} \log_\sigma
n) = o(n/r)$.

To support counting we also record in $Tbl_q[Q]$ the number of occurrences in 
$T$ of each string $Q$.

\begin{lemma}
  \label{lemma:smallram}
There exists a data structure that uses $O(n\sqrt{\log n \log\sigma})$ bits
and reports all $\occ$ occurrences of a query string $Q$ in $O(\occ)$ time if
$|Q|\le 4r+3$, with $r=\Theta(\sqrt{\log_\sigma n})$. The data structure also
computes $\occ$ in $O(1)$ time and can be built in time $O(n/\sqrt{\log_\sigma
n})$. 
\end{lemma}

\section{Conclusion}

We have described the first text index that can be built and queried in 
sublinear time. For example, on a text of length $n$ and an alphabet of 
poly-logarithmic size, the index is built in $O(n/\log^{1/2-\eps} n)$ time, 
for any constant
$0<\eps<1/2$, on a RAM machine of $\Theta(\log n)$ bits. An index that is built
in sublinear time, forcefully, must use $o(n\log n)$ bits, thus our index is
also compressed: again on a poly-logarithmic sized alphabet, it requires
$O(n\sqrt{\log n\log\log n})$ bits. There are indexes using less space, 
reaching the
asymptotically optimal $O(n\log\sigma)$ bits for an alphabet of size $\sigma$,
but those use $\Omega(n)$ construction time or $\Omega(q)$ time to count the 
occurrences of a pattern of length $q$. In contrast, ours is the first index 
using $o(n\log n)$ bits while
counting in almost-optimal time, $O(q/\log_\sigma n + \sqrt{\log_\sigma
n}\log\log n)$, and then delivering each occurrence in constant time. 

Our technique is based on {\em difference covers}, which allow us indexing
only a sublinear number of suffixes while still being able to compute the
longest common prefix between any two suffixes in constant time. Difference
covers have been using for linear-time suffix tree construction, but never for 
achieving sublinear query time.

We know no lower bound that prevents us from aiming at an index
using the least possible space, $O(n\log\sigma)$ bits, the least possible
construction time for this space in the RAM model, $O(n/\log_\sigma n)$, and
the least possible counting time, $O(q/\log_\sigma n)$. Our index is the first
one in breaking the $\Theta(n)$ construction time and $\Theta(q)$ query time
barriers simultaneously, but it is open how close 
we can get from the optimal space and construction time.

\conf{In the extended version of this paper (see 
\texttt{https://arxiv.org/abs/1712.07431})}
\arxiv{In appendix} 
we extend our model to secondary memory, with main memory size $M$ and block
size $B$. In this case, we can build the index in 
$O(\sort(n)/\sqrt{\log_{\sigma}n}+(n/B)\log\sigma)$ I/Os, where
$\sort(n)=(n/B)\log_{M/B} (n/M)$ are the I/Os needed to sort $n$ numbers, and 
report the $\occ$ occurrences of the pattern in 
$O(q/(B\log_\sigma n) + \log_B n + \sqrt{\log_{\sigma} n} \log\log n + \occ/B)$
I/Os. Alternatively, if $\log B = o(\log n /\log\log n)$,
we can build in $o(\sort(n))+O((n/B)\log\sigma)$ I/Os an index supporting 
queries in the optimal $O(q/(B\log_\sigma n) + \log_B n + \occ/B)$ I/Os.
Note that, for $\log\sigma = o(\log_{M/B} (n/M))$, 
this is the first suffix array 
construction taking $o(\sort(n))$ I/Os, which is possible because we actually
sort only the sampled suffixes. This demonstrates that, while $\sort(n)$ is a
lower bound for full suffix array construction~\cite{FFM00,KarkkainenR2003}, we can build in less time an index that emulates its search functionality and achieves  optimal or near-optimal query cost. 

\appendix

\arxiv{
\section{External-Memory Data Structure}
\label{sec:external}

In this section we extend our index to the secondary memory scenario. The fact
that we do not sort all the suffixes allows us to build the index in time
$o(\sort(n))$, which is impossible for a full suffix array. 
The main challenge is to handle, within the desired I/Os, pattern lengths that are larger than $\log^3 n$ but still not larger
than $B \log^3 n$; for longer ones the times of the sampled suffix tree search suffice.

\begin{theorem}
  \label{theor:extmem}
Given a text $T$ of length $n$ over an alphabet of size $\sigma$, in the
external memory scenario with block size $B$ and memory size $M$, there is
an index that reports the $\occ$ occurrences of a pattern of length $q$ in 
$O(q/(B\log_{\sigma}n)+ \log_B n+\sqrt{\log_{\sigma}n}\log\log n
+ \occ/B)$ I/Os and can be built in
$O(\sort(n)/\sqrt{\log_{\sigma}n}+(n/B)\cdot\log\sigma)$  I/Os, where
$\sort(n)=O((n/B)\log_{M/B}n)$ is the cost of sorting $n$ numbers.  
There is another index that, if $\log B = o(\frac{\log n}{\log\log n})$, 
reports the occurrences in the optimal
$O(q/(B\log_{\sigma}n)+ \log_B n + \occ/B)$ I/Os and is built
in $O((\frac{\log\log n}{\log_B n}+\frac{1}{\sqrt{\log_\sigma n}})\sort(n)+ (n/B)\log\sigma)$  I/Os.
\end{theorem}

Our data structure consists of the following components:
\begin{enumerate}
\item 
 We keep the suffix array and suffix tree $\cT'$ for selected suffixes.  As before, positions of selected suffixes in the text are chosen according to Lemma~\ref{lemma:dc1}; the value of parameter $r$ will be determined later. 
We also construct an inverse suffix array $SAI[0..(n/r)-1]$ for the selected suffixes:  $SAI[f]$ is the rank of the suffix starting at $T[\lfloor f/s\rfloor s +a_j]$ where $j=f\!\!\mod s$. 
\item
  We store  a data structure for range reporting queries. This is a straightforward modification  of the data structure used in our internal-memory result; see Section~\ref{sec:rangerep}. Our external range reporting data structure can be constructed in $O((n/B) \log \sigma)$ I/Os  and supports queries in $O(\log\log n + \pocc/B)$ I/Os, where $\pocc$ is the number of reported points.
\item
  We also keep another data structure that supports suffix jumps in $O(\log\log n)$ I/Os per  jump when the size of a query string does not exceed $B\log^3 n$.  This structure is the novel part of our external-memory construction and is described in Section~\ref{subsec:short}. 
\end{enumerate}

The above data structure supports pattern matching for strings that cross at
least one selected position, i.e., $|Q|\ge 4r+3$. The index for very short patterns, $|Q|< 4r+3$, is described in Appendix~\ref{sec:smallext}.

Our basic approach is the same as in the internal-memory data structure. In order to find occurrences of $Q$, we locate $Q$ in the tree $\cT'$ that contains all sampled suffixes. 
Using an external memory variant of the suffix tree~\cite{FerraginaG99}, this step can be done in $O(|Q|/(B\log_{\sigma} n) + \log_B n)$ I/Os. 
Then we execute $4r+3$ suffix jump queries and find the loci of all strings
$Q_i=Q[i..|Q|-1]$ for $1\le i < 4r+3$. For every $Q_i$ that occurs in $\cT'$
we answer a two-dimensional range reporting query in $O(\log \log n)$ I/Os,
plus $O(1/B)$ amortized I/Os per reported occurrence.  We can execute each suffix jump in $O(\log n)$ I/Os using the method of Section~\ref{sec:jumps}. This slow method  incurs an additional cost of $O(r\log n)$. However, the total query cost is not affected by suffix jumps if 
$Q$ is sufficiently large, $|Q|\ge B\cdot r\cdot \log n\log_{\sigma} n$.   We will show in Section~\ref{subsec:short} how suffix jumps can be computed in $O(r\log\log n)$ I/Os when the length of $Q$ is smaller than $B\log^3 n$.  Thus a query is answered in 
$O(q/(B\log_{\sigma} n)+ \log_B n + r\log\log n +\occ/B)$ I/Os. 
All parts of our data structure for suffix jumps can be constructed in $O((\frac{1}{r}+\frac{r}{\log_{\sigma}n})\sort(n))$ I/Os; see Section~\ref{sec:construction}. 
The data structure for range reporting can be constructed in $O((n/B)\log\sigma)$ I/Os; see Section~\ref{sec:rangerep}. The data structure for very short patterns can be constructed in $O((n/B)\log\sigma+(1/r)\sort(n))$ I/Os; see Section~\ref{sec:smallext}. Thus the total construction cost is $O((n/B)\log\sigma + (\frac{1}{r}+\frac{r}{\log_{\sigma}n})\sort(n))$. 

We can set  $r=\sqrt{\log_{\sigma}n}$ and  obtain the first result  in
Theorem~\ref{theor:extmem}. 
Alternatively, if $\log\log n = o(\log_B n)$, we can set 
$r=\min(\log_B n/\log\log n,\sqrt{\log_{\sigma}n})$, 
so that the query time is the optimal $O(q/(B \log_{\sigma} n)+ \log_B n+\occ)$.
The construction cost is as follows: $\frac{r}{\log_\sigma n}$ is at most
$\frac{1}{\sqrt{\log_\sigma n}}$, whereas $\frac{1}{r}$ is 
$\max(\frac{\log\log n}{\log_B n},\frac{1}{\sqrt{\log_\sigma n}})$, so the
cost is $O((\frac{\log\log n}{\log_B n}+\frac{1}{\sqrt{\log_\sigma n}})\sort(n))$.
This is the second result in Theorem~\ref{theor:extmem}.

\section{Suffix Jumps for Middle-Length Patterns in External Memory}
\label{subsec:short}
In this section we show how to compute suffix jumps in $O(\log\log n)$ I/Os when the query string is not too long, $|Q|\le B \log^3 n$. 

\paragraph{Overview.}
Our data structure consists of two parts.
The first part follows the method previously employed for short patterns (see  Section~\ref{sec:shortjump}) with small modifications and a different choice of parameters. The set $\cS_2$ is obtained by selecting every $(B^3\log^{10} n)$-th suffix from $\cS$.  We define $\cY_0$ and $\cY$ analogously to $\cX_0$ and $\cX$ in Section~\ref{sec:shortjump}: for each suffix $T[i..]$ in $\cS_2$, $\cY_0$ contains all substrings $T[i..i+f]$ for $0\le f\le B\log^3n$ and $\cY$ contains all substrings $T[i+f_1..i+f_2]$ for $0\le f_1\le 4r+3$ and $0\le f_2\le B\log^3 n$.  Thus $\cY_0$ contains all prefixes of $T[i..i+B\log^3 n]$ and $\cY$ contains all substrings obtained by removing up to $4r+3$ leading symbols from a string in $\cY_0$.  Sets $\cY$ and $\cY_0$ are similar to sets 
$\cX$ and $\cX_0$.  

 \begin{lemma}
\label{lemma:jumpy0}
   For any string $Q$, such that $4r+3<|Q|\le B\log^3 n$, and for all $i$, $1\le i\le 4r+3$,  we can find $LCP(Q[i..],\cY_0)$, the locus of $LCP(Q[i..],\cY_0)$ in $\cT'$, and $|LCP(Q[i..],\cY_0)|$ in $O(\log_B n + |Q|/(B\log_{\sigma}n)+r\log\log n)$ I/Os. 
 \end{lemma}
 \begin{proof}
We keep all strings from $\cY$ in the string B-tree $\cT_Y$. We also keep all strings from $\cY_0$ in a trie $\cT_0$. 
Let $Q'_i=LCP(Q[i..|Q|-1],\cY_0)$, let $u_i$ be the locus of $Q'_i$ in $\cT_0$, and let $w_i$ be the parent node of $u_i$. Let $g=4r+3$. 

We keep a look-up table for  prefixes of  strings in $\cY_0$. This table contains the strings $P$ of length up to $g$ and their loci in $\cT$ for all strings $P$ that are prefixes of strings in $\cY_0$. Since there are $\sigma^{g+1}=o(\sqrt{n})$ different strings of length at most $g$, we can initialize the table in $o(\sqrt{n}/B)$ I/Os. First we   determine whether $|Q'_i|> g$; if $|Q'_i|\le g$, we  can   compute $|LCP(Q,\cY_0)|$ and $u_i$. Using the above described look-up table, this can be done in $O(g)$ I/Os. 
Then we find the $LCP(Q_g,\cY)$ and its locus $v$ in $\cT_Y$. We can find $v$   in $O(|Q|/(B\log_{\sigma}n)+\log_B n)$ I/Os by searching in the string B-tree of $\cY$~\cite{FerraginaG99}. For each $i$ satisfying $|Q'_i|\ge g$,  we find the node $w_i$ in $\cT_0$ using auxiliary data structures that will be described below. Since $u_i$ is a child of $w_i$, we can find each $u_i$ in $O(1)$ I/Os when $w_i$ is known. When we found $u_i$, $|Q'_i|$ can be computed in $O(1)$ I/Os.

Now we describe  how $w_i$ can be found if $|Q'_i|\ge g$ and the node $v$ is known. We will say that a node $v$ of $T_S$ is marked with a string $s$ if $v$ corresponds to some string $S_v$ and there is a  node $w\in \cT_0$  that corresponds to $s\circ S_v$, the concatenation of $s$ and $S_v$. In this case we also keep a pointer from $v$ to $w$. We keep the data structure that answers  marked ancestor queries: for any node $v$ and for any  string $s$ of length at most $g$, find the lowest ancestor of $v$ that is marked with $s$. In Appendix~\ref{sec:markanc} we describe a data structure that supports  marked ancestor queries in $O(\log\log n)$ I/Os provided that  the string depth of all nodes does not exceed $B\log^3 n$.  While a data structure with the same query cost was described previously~\cite{MuthukrishnanM96}, our data structure can be efficiently constructed in the external memory model.

In order to find the node  $w_i$, we look for the lowest ancestor $w'_i$ of $v$ that is marked with $Q[i..g]$. Then we follow the pointer from $w_i'$ to the corresponding node in $\cT_0$. 
Correctness of this procedure can be shown as follows. 
If a string $P[i..|P|-i-1]$, for $0\le i \le g$ and $|P|>g$, is a prefix of some string in $\cY_0$, then $P[g..|P|-g-1]$ is a prefix of some string in $\cY$. This statement can be also extended to tree nodes: if a node $w$ in $\cT_0$ corresponds to a string $P_i=P[i..|P|-i-1]$, then there are two strings in $\cY_0$ that start with $P_i$ and differ in $(|P_i|+1)$-st symbol. Hence for every node $w$ in $\cT_0$ that corresponds to $P_i$ there is a node in $\cT_Y$ that corresponds to  $P[g..|P|-g+i-1]$ and is marked with $P[i..g-1]$. Hence, if $w'_i$ is the lowest ancestor of $v$ marked with $Q[i..g]$, then the node $w_i$ is the lowest node that corresponds to a prefix of $Q[i..]$. 
 \end{proof}

The second part of our construction  is a data structure for searching in a set of $B^3 \log^{O(1)}n$ suffixes from $\cS'$. This structure, described in Lemmas~\ref{lemma:lcpB1} and \ref{lemma:lcpB2}, enables us to finish the search for the LCP by looking among   leaf descendants of $LCP_B(Q[i..],\cY_0)$. First, we show in Lemma~\ref{lemma:lcpBsmall} how a suffix jump  on $B^3\log^{O(1)} n$ suffixes can be implemented when $B$ is bounded by $s\log^3 n$, where $s$ is the block size. This case is easy because we have to search among a poly-logarithmic number of suffixes; hence we can use the same method as in Lemmas~\ref{lemma:longjump} and \ref{lemma:jumpshort}. The case when $B\ge s\log^3 n$ is more challenging: we cannot afford to spend $O(\log B)$ I/Os on a suffix jump when $B$ is sufficiently large. Our structure for this case will be a conceptual B-tree with arity $B/\log^3n$ and with leaves spanning $\log^{O(1)} n$ suffixes. In Lemma~\ref{lemma:lcpB1} we show how to compute a suffix jump on $B/\log^3n$ suffixes. Since the number of suffixes is small relative to the block size, we can keep all relevant information in one block. We can then search among $B\log^{O(1)}n$ suffixes by traversing $O(\log\log n)$ nodes in a root-to-leaf path in this conceptual tree, and then finishing the search on the leaves, in another $O(\log\log n)$ I/Os; see Lemma~\ref{lemma:lcpB2}. In Lemmas~\ref{lemma:lcpB1} and~\ref{lemma:lcpB2} we make an additional assumption that the LCPs between the suffix and the set are $>s$.  We get rid of this assumption and obtain the final result in Lemma~\ref{lemma:extmemjump}.

\paragraph{LCP Queries on a Set of $B^3\log^{10}n$ Suffixes.}
We start with the case when $B$ is small. 
\begin{lemma}
  \label{lemma:lcpBsmall}
Suppose that a subtree $\cT_v$ has $O(B^3\log^{10} n)$ leaves and let $\cV$ denote a set of   suffixes in $\cT_v$. If $B< s\log^3n$, then we can find $LCP(\Suf[i..],\cV)$ for any suffix $Suf$ from $\cS$ for any $i\le 4r+3$ in $O(\log\log n)$ I/Os with a data structure that uses $O(|\cV|)$ words of space. 
\end{lemma}
\begin{proof}
  If $B<s\log^3n$, then $\cV$ contains $O(s^3\log^{19} n)=O(\log^{22} n)$ consecutive suffixes. We can find $LCP(\Suf[i..],\cV)$ in $O(\log\log n)$ I/Os by applying Lemma~\ref{lemma:lcp} $O(\log\log n)$ times as described in the second part of Lemma~\ref{lemma:jumpshort}. 
\end{proof}

Now we consider the situation when $B$ is large. 
\begin{lemma}
  \label{lemma:lcpB1}
Let $\cV\subset \cS'$ be a set of $O(B/(s\log^3 n))$  suffixes. Suppose that
$|LCP(\Suf,S)|> s$ for some selected suffix $\Suf \in \cS'$ and for all
suffixes $S$ from $\cV$. Then for all  $i$, $0\le i<4r+3$, we can find
$LCP(\Suf[i..],\cV)$ in $O(1)$ I/Os with a data structure that uses
$O(|\cV|\cdot s\cdot \log^3n)$ space. 
\end{lemma}
\begin{proof}
The first step is to find the lexicographic position of $\Suf'=\Suf[i..]$ in 
$\cV$. We define $\rank(\Suf')$ as its position in the suffix array of $T$, 
i.e., its position in the lexicographic order of all suffixes from $\cS$.
We would just need then to find the position of $\rank(\Suf')$ among the ranks
of the suffixes in $\cV$.
To compute $\rank(\Suf')$ we may use our sampled inverse suffix array, $SAI$.
However, $SAI$ is defined for selected suffixes only. Our plan is to {\em shift}
the suffix $\Suf'$ forward until finding a selected suffix. The resulting rank
is very different from that of $\Suf'$, but it can still be compared with the
ranks of other similarly shifted suffixes that share the first characters with
$\Suf'$.

Let $\Suf = T[f-i..]$, so that $\Suf' = \Suf[i..] = T[f..]$. Let $f = gs+k$
for some $0 \le k < s$. Say that an $(a_j)$-suffix of $\cV$ is one of the form
$T[ls + a_j..]$. Since $\Suf$ and suffixes in $\cV$ share their first $s$ symbols, the rank
of $\Suf'$ among the $(a_j)$-suffixes of $\cV$ stays the same if we shift all
the suffixes by $h(k,a_j)$. Now $\Suf'[h(k,a_j)..] = 
T[f+h(k,a_j)..]=T[gs+k+h(k,a_j)..]$ is a selected suffix and we can find 
$x_j = \rank(\Suf'[h(k,a_j)..])$ using $SAI$. We denote this {\em shifted rank}
as $\shrank(\Suf',k,j) = \shrank(T[f..],k,j)=\rank(T[f+h(k,a_j)..])$. Note, on
the other hand, that an $(a_j)$-suffix shifted by $h(k,a_j)$ is still a selected
suffix, and thus its rank is known in our index.

For each $0\le k < s$ and $0 \le j \le 4r+3$, we create a data structure
$Z_{k,j}$ containing
all the shifted ranks of $(a_j)$-suffixes of $\cV$ so that they can be compared 
with a suffix of the form $T[gs+k..]$. That is, $Z_{k,j} = \{ \shrank(S,k,j),~ S \in \cV \textrm{~is an~}(a_j)\textrm{-suffix} \}$. Shifted ranks stored in 
$Z_{k,j}$ can then be used to find the LCP of a suffix of the form $T[gs+k ..]$
among $(a_j)$-suffixes: the closest suffixes are 
$prev_j= \pred(x_j, Z_{k,j})$ and $next_j=\newsucc(x_j,Z_{k,j})$.

However, we cannot compare ranks of suffixes shifted by different values
$h(k,a_j)$. Moreover, $prev_j$ and $next_j$ for different values of $j$ can be stored in different parts of the sampled suffix array. In order to find the LCP between all $prev_j$ ($next_j$) and the shifted suffix $Suf'$, we need to answer $4r+3$ queries on different parts of the suffix array. This would increase the cost significantly. 
We will use a different mechanism to obtain an approximation of 
$LCP(\Suf',\cV)$ with additive error $B$, and then complete the calculation. 
With this goal in mind, for every suffix $S \in \cV$ whose shifted rank 
$\shrank(S,k,j)$ is stored in some $Z_{k,j}$ we will store $2+4\log^3 n$ 
{\em marginal values} associated with $S'=S[h(k,a_j)..]$: for 
$i=0,1,\ldots, 2\log^3 n-1$, the marginal value $\pmargin(S',i)$ is the rank 
of the leftmost suffix $S_i< S'$ such that
$(B/2)\cdot i \le |LCP(S_i,S')|< (B/2)(i+1)$; for $i=2\log^3n$, $\pmargin(S',i)$
is the rank of the leftmost suffix $S_i< S'$ such that $|LCP(S_i,S')|\ge B\log^3
n$. Similarly, the marginal value $\nmargin(S',i)$ is the rank of the rightmost
suffix $S_i> S'$ such that $(B/2)\cdot i \le |LCP(S_i,S')|< (B/2)(i+1)$ if
$i<2\log^3 n$, and the rank of the rightmost suffix $S_i> S'$ such that 
$|LCP(S_i,S')|\ge B\log^3 n$ if $i=2\log^3 n$. To complete the search, we will 
also use a string B-tree $\cT_{\cV}$ containing the suffixes in $\cV$.

Consider suffix $T[f..]$ and the suffix $S = T[f'..]$ with shifted rank 
$next_j \in Z_{k,j}$, thus $\rank(T[f..]) \le \rank(T[f'..])$ and
$\shrank(T[f..],k,j)\le \shrank(T[f'..],k,j)$. Let $S' = S[h(k,a_j)..] =
T[f'+h(k,a_j)..]$, and let $i$ be such that $\pmargin(S',i) \le 
\shrank(T[f..],k,j) < \pmargin(S',i+1)$. Then 
$\ell=LCP(T[f+h(k,a_j)..],T[f'+h(k,a_j)..])$ satisfies $(B/2)i\le \ell <
(B/2)(i+1)$. Since $h(k,a_j)\le s$ and the first $s$ symbols in $T[f..]$ and
$T[f'..]$ are equal, we have $(B/2)i \le |LCP(T[f..],T[f'..])|\le
(B/2)(i+1)+s\le (B/2)i + B$. Hence data structures $Z_{k,j}$ with  marginal
values provide an estimate for  the $LCP$ of any suffix and the set of
$(a_j)$-suffixes of $\cV$. The desired value $LCP(\Suf[i..],\cV)$ is the 
maximum LCP value over all the sets of $(a_j)$-suffixes in $\cV$. Note that 
the index $i$ is smaller than $2\log^3n$:
if $i=2\log^3 n$, then $\ell\ge B \log^3 n$ and $Q$ is a long pattern.
Estimating the LCP with respect to $prev_j$ is analogous. Since all the data 
structures we have described fit into one block, this part of the query takes 
only $O(1)$ I/Os. 

Now we summarize our method and describe the complete procedure to answer an
LCP query for $\Suf'=\Suf[i..]$ for some selected suffix $\Suf=T[f-i..]=
T[gs+a_t]$. Let $pos(f-i)=g+t$ denote the position corresponding to $f-i$ in 
$SAI$. Since we assume $s<B$, we can read $SAI[pos(f-i)]$, $SAI[pos(f-i)+1]$, 
$\ldots$, $SAI[pos(f-i)+s]$ into main memory in $O(1)$ I/Os. We also read the 
data structures for $\cV$ into main memory. For each $0\le j\le 4r+3$, we 
compute $x_j=SAI[pos((f-i)+h(a_t+i,a_j))]$ and find $prev_j=\pred(v_j, Z_{k,j})$
and $next_j=\newsucc(v_j,Z_{k,j})$. We then estimate $LCP(\Suf',prev_j)$ and 
$LCP(\Suf',next_j)$ using marginal values as described above. Let $S$ denote 
the suffix that provides the longest estimated LCP. Suffix $S$
provides an approximation for $LCP(\Suf',\cV)$ with an additive error of 
at most $B$. Let $u$ denote the locus of $LCP(S,\cV)$ in the string B-tree 
$\cT_{\cV}$. Node $u$ can be found in $O(\log \log n)$ I/Os using a weighted 
level ancestor query~\cite{AmirLLS07} from the leaf that contains the suffix 
$S$. We can then move 
down from $u$ by at most $B$ symbols and identify $LCP(\Suf',\cV)$ in $O(1)$ 
I/Os using the standard search procedure in a string B-tree~\cite{FerraginaG99}:
For a node $u$ of size $|u| = O(B^3 \log^{10} n)$, the search cost is 
$O(B/B + \log_B |u|)=O(1)$, since $B \ge s\log^3 n$. 

Since $\sum_{j} |Z_{k,j}| = |\cV|$, the total space is dominated by the 
$O(\log^3 n)$ marginal values stored for every element of every $Z_{k,j}$,
adding over the $s$ values of $k$. This gives a total of $O(|\cV|\cdot s \cdot
\log^3 n)$.
\end{proof}

Now we will show how the result of Lemma~\ref{lemma:lcpB1} can be extended to a set with $O(B^3\mathrm{polylog}(n))$ suffixes and finish the description of suffix jumps on middle-length patterns. 

\begin{lemma}
  \label{lemma:lcpB2}
Let $\cV$ denote a set of $O(B^3\log^{10} n)$ consecutive  suffixes in $\cT'$.
Suppose that, for some selected suffix $\Suf \in \cS$, some $0 \le i\le 4r+3$,
and for all suffixes $S \in \cV$, it holds $|LCP(\Suf[i..],S)|> s$. Then  we can find $LCP(\Suf[i..],\cV)$ in $O(\log\log n)$ I/Os with a structure that uses $O(|\cV|)$ space. 
\end{lemma}
\begin{proof}
  We assign all suffixes of $\cV$ to nodes of a conceptual tree $\cTA$. Every leaf of $\cTA$ stores $s^3 \log^{9}n$ suffixes and every internal node has $B/(s\log^3 n)$ children.  We associate a set $\cA(\nu)$ with every node $\nu$ of $\cTA$; $\cA(\nu)$ contains two representative suffixes, i.e., the smallest and the largest suffix,  from every child of $\nu$. 
  
  In order to find $LCP(\Suf,\cV)$ we start at the root node of $\cTA$ and move down to the leaf.  In every internal node $\nu$, we can find the child that contains $LCP(\Suf,\cA(\nu))$ in $O(1)$ I/Os using Lemma~\ref{lemma:lcpB1}. When we reach a leaf node $\ell$, we can finish the search in $O(\log\log n)$ I/Os: Let $v$ denote the lowest common ancestor of all suffixes from $\ell$ in the suffix tree $\cT'$ and let $v_1$ denote the child of $v$ that contains $LCP(\Suf[i..],\cV)$. The subtree $\cT_v$ rooted at $v_1$ has at most $O(\log^{12} n)$ leaves. 
Since we know that the locus of $LCP(\Suf[i..],\cV)$ is in $\cT_v$, we can find it in $O(\log\log n)$ I/Os by applying Lemma~\ref{lemma:lcp} $O(\log\log n)$ times; see the second part of Lemma~\ref{lemma:jumpshort}.

The total number of internal nodes in $\cTA$ is $O(|\cV|/(s^3\log^9 n))$. Hence the total number of suffixes in all $\cA(\nu)$ is also bounded by $O(|\cV|/(s^3\log^9 n))$. Since each data structure for $\cA(\nu)$ needs $O(|\cA(\nu)|s\log^3 n)$ space, the total space usage is bounded by $O(|\cV|)$. 
\end{proof}

\begin{lemma}
\label{lemma:extmemjump}
 Suppose that we know $LCP(Q[i..q-1],\cY_0)$ and their loci in $\cT'$ for all
$0 \le i\le 4r+3$. Suppose that we also know $Q'=LCP(Q[j..q-1],\cS')$ and its locus in $\cT'$ for some $0\le j\le 4r+3$. Let $Q_f=Q[f..q'-1]$ where $q'=|Q'|$ and $f\ge j$.  If $|LCP(Q[f..q'-1],\cY_0)|\ge \log^3 n$, then we can compute $Q'_f=LCP(Q_f,\cS')$  and its locus in $\cT'$ in $O(\log \log n)$ I/Os for any $f\le 4r+3$. 
\end{lemma}
Lemma~\ref{lemma:extmemjump} is proved in the same way as Lemma~\ref{lemma:jumpshort}. Either $LCP(Q_f,\cS')=LCP(Q_f,\cY_0)$ or we can identify a subtree $\cT_v$ of $\cT'$, such that 
$\cT_v$ has $O(B^3\log^{10}n)$ leaves and $LCP(Q_f,\cS')$ is in $\cT_v$. Using
Lemma~\ref{lemma:lcpB2}, we can then find $LCP(Q_f,\cS')$ in $\cT_v$ with
$O(\log\log n)$ I/Os. 

\paragraph{Putting All Parts Together.}
Now we are ready to describe suffix jumps on middle-length patterns. One additional component is needed to compute suffix jumps when the query pattern (respectively its LCP) is short. We keep the same data structure as described above that supports  suffix jumps for short query strings, $|Q|\le \log^3 n$. 
That is, we keep the set $\cX_0$ and a data structure on $\cX_0$ equivalent to the structure of Lemma~\ref{lemma:jumpy0}. 

For a query pattern $Q[0..|q|]$, we proceed as follows. Let $l=\min(q,\log^3 n)$. We start by identifying $LCP(Q[i..l],\cX_0)$  and their locus nodes $\nu_i$ for all $0\le i\le 4r+3$ as described in Lemma~\ref{lemma:jumpy0}.  If $q\le \log^3 n$, we complete the search as explained in Section~\ref{sec:shortjump}.  If $q> \log^3 n$, we also find $LCP(Q[i..q],\cY_0)$ and their loci $\mu_i$ using Lemma~\ref{lemma:jumpy0}. 

Suppose that $Q'=LCP(Q[j..q-1],\cS')$ and its locus in $\cT'$ are known and we want to compute $Q_f=LCP(Q'[f..q'-1],\cS')$ for $q'=|Q'|$ and for some $f$ satisfying $j< f\le 4r+3$.
If $q\le \log^3 n$, we finish the search as explained in Section~\ref{sec:shortjump}. 
Suppose that $q> \log^3 n$.  If the depth of $\nu_f$ is smaller than $\log^3 n$, then consider   the child $\nu'_f$ of $\nu_f$ that contains the locus of $Q_f$. The node  $\nu'_f$  has $O(\log^{10} n)$ leaf descendants: recall that  $\cX_0$ contains substrings corresponding to suffixes from $\cS''$ and $\cS''$ contains every $\log^{10}n$-th suffix from $\cS'$. Hence we can also complete the search as described in Section~\ref{sec:shortjump}.  Otherwise let $\rho_f$ denote the lowest of the two nodes, $\mu_f$ and $\nu_f$. If $\rho_f$ is not the locus node, let $\rho'_f$ be the child of $\rho_f$ that contains $Q_f$. The depth of $\rho'_f$ is at least $\log^3 n$ and it has $O(B^3\log^{10}n)$ leaf descendants because $\rho'_f$  is below $\mu_f$. Hence we can find $Q_f$ and its locus node  in $O(\log\log n)$ I/Os using either Lemma~\ref{lemma:lcpBsmall} or Lemmas~\ref{lemma:lcpB1} and~\ref{lemma:lcpB2}.
\begin{lemma}
  \label{lemma:shortlociB}
Suppose that $|Q|\le B\log^3 n$. In $O(r\log\log n +|Q|/(B\log_{\sigma} n))$ I/Os we  can find all existing loci of $Q[i..|Q|]$, $0\le i < 4r+3$, in $\cT'$.
\end{lemma}


\section{Construction Algorithms: External Memory}
\label{sec:consextmem}
We consider the set of strings $P_{ij}=T[si+a_j..s(i+1)+a_j-1]$, i.e., all strings of length $s$ that start at selected positions. We view $P_{ij}$ as strings of length $O(s/\log_{\sigma}n)$ over an alphabet of size $\sigma^{\rho}$ for some $\rho$, $(1/2)\log_{\sigma}n\le \rho< \log_{\sigma}n$. In other words, we view $P_{ij}$ as sequences of meta-symbols, such that each meta-symbol represents $\rho=\Theta(\log_{\sigma}n)$ symbols and each meta-symbol fits into one word. Since $s=\Theta(r^2)$, we can sort all strings $P_{ij}$ in $O(\frac{nr}{B\log_{\sigma}n}\log_{M/B}n)=O((r/\log_{\sigma}n)\sort(n))$ I/Os~\cite{ArgeFGV97}.

We assign lexicographic names to every string and construct texts $T_k= [t_{a_k}...t_{a_k+s-1}][t_{a_k+s}...t_{a_k+2s-1}]...$ for $k=0$, $\ldots$, $4r+3$. That is, the $i$-th symbol in $T_k$ is the lexicographic name of the string $T[a_k+is..a_k+(i+1)s-1]$. 
Let $\oT$ denote the concatenation of texts $T_k$, $\oT=T_1\$T_2\$\ldots T_{4r+3}\$ $. Since $\oT$ consists of $O(n/r)$ symbols, we construct a suffix array, a suffix tree, and an inverse suffix array $SAI$  for sampled suffixes in $O(\sort(n/r))$ I/Os~\cite{KSB06,FFM00}.  We can also construct the string B-tree on suffixes within the same time bounds~\cite{FFM00}.  To support suffix jumps, we also  need a data structure that answers weighted level ancestor queries.  A construction algorithm  for this data structure is described in Section~\ref{sec:wlacons}. This algorithm runs in $O(\sort(n/r))$ I/Os. 

Now we show how external data structures for suffix jumps are generated. Recall that  $\cY$ is the set of strings   $T[i+f_1..i+f_2]$, $0\le f_1\le 4r+3$ and $0\le f_2\le B\log^3 n$, such that $T[i..]\in \cS_1$. $\cY$ contains $O(n/\log^4n)$ strings of total length $O(n/\log n)$. We can generate all strings from $\cY$ in $O(n/(B\log n))$ I/Os: first, we identify the starting positions of strings in $\cS_1$; since the suffix array is already constructed we can find the indexes $i$, such that $T[i..]\in \cS'$ in $O(n/(Br))$ I/Os. Then we traverse the text $T$ and generate strings  $T[i+f_1..i+f_2]$ in $O(n/(B\log n))$ I/Os.  When strings are generated we insert all strings from $\cY_0$ into $\cT_0$ and all strings from $\cY$ into $\cT_Y$. Next, we mark strings in $\cT_0$. For every  node $u$ of $\cT_Y$ we generate a string corresponding to $u$. We store these strings in a set $L_Y$; for each string a pointer to the corresponding node $u\in \cT_Y$ is also stored. We also collect strings corresponding to nodes of $\cY_0$, sort them, and store them in the list $L'$. For each string in $L'$  we also store the corresponding node of $\cT_0$. For every $j$, $0\le j\le 4r+3$, we proceed as follows. We remove $j$ leading symbols from every string in $L_Y$ and sort the resulting strings. Let $L_j$ be the list of lexicographically sorted strings from $L_Y$ with $j$ first symbols removed. We simultaneously traverse $L'$ with $L_j$: if some string $s$ occurs in both lists, we add $s$ to the list $M_j$. For each string in $M_j$, we also keep the index of the corresponding node in $\cT_0$ and the string $s^t$ of removed symbols. Then we sort the list $M_j$ according to the in-order of nodes in $\cT_0$. Now we know with what strings nodes of $\cT_0$ must be marked. Finally we traverse the tree $\cT_0$ and the list $M_j$. When we visit a node $u$ of  $\cT_0$ that corresponds to a string $s$, we find $s$ in the list $M_j$ and mark $u$ with $s^t$. When we performed this procedure for all $j$, $0\le j\le 4r+3$, all nodes of $\cT_0$ are marked correctly. Finally we construct the marked ancestor data structure in $O(\sort(n/r))$ I/Os; see Sections~\ref{sec:markanc} and~\ref{sec:wlacons}.  Data structures for sets $\cX$ and $\cX_0$ can be created in the same way.  The total construction cost for data structures supporting suffix jumps is bounded by the cost of sorting strings in $\cY$ (resp. the cost of sorting strings in $\cX$). Hence data structures for suffix jumps can be constructed in $O(\sort(n/r))$ I/Os.

In order to construct data structures of Lemma~\ref{lemma:lcpB1} and Lemma~\ref{lemma:lcpB2} we need to know shifted ranks of some suffixes, i.e., we must know the rank of $T[f+h(i,j)..]$ for every $T[f..]$ in $\cV$. This information can be also obtained by careful  sorting and extracting data from the inverse suffix array. We traverse the suffix array of sampled suffixes and mark every $(s\log^3 n)$-th suffix. Next we sort marked suffixes by their starting positions in the text. Let $L_m$ denote the list of marked suffixes.  We traverse the inverse suffix array $SAI$; for every position $SAI[x]$ that corresponds to a marked suffix $T[f..]$, we add the values of $SAI[x+1]$, $\ldots$, $SAI[x+s]$ to the entry for $T[f..]$. Next we again sort the entries of $L_m$ (with added information) by their ranks.  Since we now know the shifted ranks of marked suffixes, we can construct data structures for sets $\cV$. The construction cost is dominated by the cost of sorting strings $P_{ij}$ and the cost of constructing the suffix array for $\oT$. Hence our data structure can be constructed in $O((1/r)\sort(n)+ (r/\log_{\sigma}n)\sort(n))$ I/Os. 

}

\section{Range Reporting}
\label{sec:rangerep}

In this section we prove the following result on two-dimensional orthogonal range reporting queries.

\begin{theorem}
\label{thm:ramrep}
  For a set of $t=O(n/r)$ points on a $t\times \sigma^{O(r)}$  grid, where $r=O(\sqrt{\log_{\sigma} n})$, and for any constant $0<\eps<1/2$, there is an $O(n\log\sigma \log^\eps n)$-bit data structure that can be built in $O(n\log^{3/2}\sigma/\log^{1/2-\eps}n)$ time and supports orthogonal  range reporting queries in time $O(\log\log t +\pocc)$ where $\pocc$ is the number of reported points.
\end{theorem}

Our method builds upon  the recent previous work on wavelet tree construction~\cite{MNV16,BabenkoGKS15}, and applications of wavelet trees to range predecessor queries~\cite{BelazzouguiP16}, as well as on results for compact range reporting~\cite{Chazelle88,ChanLP11}. 

\subsection{Base data structure}

We are given a set $\cQ$ of $t=O(n/r)$ points in $[0..t-1] \times [0..\sigma^{O(r)}]$. First we sort the points by $x$-coordinates (this is easily done by scanning the leaves of $\cT'$, which are already sorted lexicographically by the selected suffixes), and keep the $y$-coordinates of every point in a sequence $Y$.  Each element of $Y$ can be regarded as a string of length $O(r)=O(\sqrt{\log_{\sigma} n})$, or equivalently, a number of $h=O(r\log\sigma)=O(\sqrt{\log n\log\sigma})$ bits.  Next we construct the range  tree for $Y$ using a method similar to the wavelet tree \cite{GGV03} construction algorithm. 
Let $Y(u_o)=Y$ for the root node $u_o$.  We classify the elements of $Y(u_o)$ according to their highest bit and generate the corresponding subsequences of $Y(u_o)$, $Y(u_l)$ (highest bit zero) and $Y(u_r)$ (highest bit one), that must be stored in the left and right children of $u$, $u_l$ and $u_r$, respectively. Then  nodes $u_l$ and $u_r$ are recursively processed in the same manner.  When we generate the sequence for a node $u$ of depth $d$, we assign elements to $Y(u_l)$ and $Y(u_r)$ according to their $d$-th highest bit.  The total time needed to assign elements to nodes and generate $Y(u)$, using a recent method \cite{MNV16,BabenkoGKS15}, is 
$O(t \cdot h/\sqrt{\log t}) = O(n\log \sigma/\sqrt{\log n})$, and the space is 
$O(t \cdot h) = O(n\log\sigma)$ bits.  

For every sequence $Y(u)$ we also construct an auxiliary data structure that supports three-sided queries. If $u$ is a right child, we create a data structure that returns all elements  in a range $[x_1,x_2]\times [0,h]$ stored in $Y(u)$.  To this end, we divide $Y(u)$ into groups $G_i(u)$ of $g=(1/2)\sqrt{\log n\log\sigma}$ consecutive elements (the last group may contain up to $2g$ elements). 
Let $\min_i(u)$ denote the smallest element in every group and let $Y'(u)$ denote 
the sequence of all $\min_i(u)$. We construct a data structure that supports three-sided queries on $Y'(u)$; it uses  $O(|Y'(u)|\log n) = O((|Y(u)|/g)\log n)=O(|Y(u)|\sqrt{\log n/\log\sigma})$ bits and reports the $k$ output points in $O(\log\log n + k)$ time; we can use any  range minimum data structure for this purpose \cite{BenderF00}. 
We can traverse $Y(u)$ and identify the smallest element in each group in $O(|Y(u)| h / \log n) = O(|Y(u)|\sqrt{\log\sigma/\log n})$ time, by using small precomputed tables that process $(\log n)/2$ bits in constant time. This adds up to
$O(t \cdot h \sqrt{\log\sigma/\log n}) = O(n\log^{3/2}\sigma/\sqrt{\log n})$ 
time.

Since the number of points in $Y'(u)$ is $O(|Y(u)|/g)$, the data structure for $Y'(u)$ can be created in $O(|Y(u)|/g)$ time and uses $O((|Y(u)|/g)h)=O(|Y(u)|)$ bits, which adds up to $O(n\sqrt{\log\sigma/\log n})$ construction time and
$O(n\log\sigma)$ bits of space.

In order to save space, we do not store the $y$-coordinates of points in a group. The $y$-coordinate of each point in $G=G_i(u)$ is replaced with its rank, that is, with the number of points in $G$ that have smaller $y$-coordinates. Each group $G$ is divided into $(\log \sigma)/(2\log\log n)$ subgroups, so that each subgroup contains $2r\log\log n$ consecutive points from $G$. We keep the rank of the smallest point from each subgroup of $G$ in a sequence $G^t$. Since the ranks of points in a group are bounded by $g$ and thus can be encoded with $\log g \le \log\log n$ bits, each subgroup can be encoded with less than $2r(\log\log n)^2$ bits.  Hence we can store precomputed answers to all possible range minimum queries on all possible subgroups in a universal table of size $O(2^{2r(\log\log n)^2}\log^2 g)=o(n)$ bits.  We can also store pre-computed answers for range minima queries on $G^t$ using another small universal table: $G^t$ is of length $(\log\sigma)/(2\log\log n)$ and the rank of each minimum is at most $g$, so 
$G^t$ can be encoded in at most $(\log\sigma)/2$ bits. This second universal table is then of size $O(2^{(\log\sigma)/2}\log^2 g)=o(n)$ bits.

A three-sided query $[x_1,x_2]\times [0,y]$  on a group $G$ can then be answered as follows. We identify the point of smallest rank in $[x_1,x_2]$. This can be achieved with  $O(1)$ table look-ups because a query on $G$ can be reduced to one query on $G^t$ plus a constant number of queries on sub-groups. Let $x'$ denote the position of this smallest-rank point in $Y(u)$. We obtain the real $y$-coordinate of $Y(u)[x']$ using the translation method that will be described below. If the real $y$-coordinate of $Y(u)[x']$  does not exceed $y$, we report it and recursively  answer three-sided queries $[x_1,x'-1]\times [0,y]$ and $[x'+1,x_2]\times [0,y]$. The procedure continues until all points in $[x_1,x_2]\times [0,y]$ are reported.  

If $u$ is a left child, we use the same method to construct the data structure  that  returns all elements  in a range $[x_1,x_2]\times [y,+\infty)$ from $Y(u)$. 

An orthogonal range reporting query $[x_1,x_2]\times [y_1,y_2]$ is then answered by finding the lowest common ancestor $v$ of the leaves that hold $y_1$ and $y_2$. Then we visit the right child $v_r$ of $v$, identify the range $[x'_1,x_2']$ and report all points in $Y(v_r)[x_1'..x_2']$  with $y$-coordinates that do not exceed $y_2$; here $x'_1$ is the index of the smallest $x$-coordinate in $Y(v_r)$ that is $\ge x_1$ and  $x_2'$ is the index of the largest $x$-coordinate of $Y(v_r)$ that is $\le x_2$. We also visit the left child $v_l$ of $v$, and answer the symmetric three-sided query. Finding $x_1'$ and $x_2'$ requires predecessor and successor queries on $x$-coordinates of any $Y(v_r)$; the needed data structures are described in Section~\ref{sec:predsuc}.

In total, the basic part of the data structure requires $O(n\log\sigma)$ bits of space and is built in time $O(n\log^{3/2}\sigma/\sqrt{\log n})$.

\subsection{Translating the answers}

An answer to our  three-sided query returns positions in $Y(v_l)$ (resp.\ in $Y(v_r)$).   We need an additional data structure to translate such local positions  into the points to be reported. While our range  tree can be used for this purpose, the cost of decoding every point would be $O(\sqrt{\log\sigma\log n})$. A faster decoding method~\cite{Chazelle88,Nekrich09,ChanLP11} enables us to decode each point in $O(1)$ time. Below we describe how this decoding structure can be constructed within the desired time bounds.

Let us choose a constant $0 < \eps < 1/2$ and,
to simplify the description, assume that $\log_{\sigma}^{\eps}n$ and $\log \sigma$ are integers. We will say that a node $u$ is an $x$-node if the height of $u$ is divisible by $x$. For an integer $x$ the  $x$-ancestor of a node $v$ is the lowest ancestor $w$ of $v$, such that  $w$ is an $x$-node. 
Let $d_k=h^{k\eps}$ for $k=0,1,\ldots, \ceil{1/\eps}$. We construct sequences $\up(u)$ in all nodes $u$.   $\up(u)$ enables us to move from a $d_k$-node to its $d_{k+1}$-ancestor: Let $k$ be the largest integer such that $u$ is a $d_k$-node and let $v$ be the $d_{k+1}$-ancestor of $u$. We say that $Y(u)[i]$ corresponds to $Y(v)[j]$ if  $Y(u)[i]$ and $Y(v)[j]$
represent the $y$-coordinates of the same point. Suppose that a three-sided query has returned position $i$ in $Y(u)$. Using auxiliary structures, we find the corresponding position $i_1$ in the $d_1$-ancestor $u_1$ of $u$. Then we find $i_2$ that corresponds to $i_1$ in the $d_2$-ancestor $u_2$ of $u_1$. We continue in the same manner, at the $k$-th step moving from a $d_k$-node to its $d_{k+1}$-ancestor. After $O(1/\eps)$ steps we reach the root node of the range tree.

It remains to describe the auxiliary data structures. To navigate from a node $v$ to its ancestor $u$, $v$ stores for every $i$ in $Y(v)$ the corresponding position $i'$ in $Y(u)$ (i.e., $Y(v)[i]$ and $Y(u)[i']$ are $y$-coordinates of the same point).  In order to speed up the construction time, we store this information in two sequences. The sequence $Y(u)$ is divided into chunks; if $u$ is a $d_k$-node, then the size of the chunk is $\Theta(2^{d_k})$.  For every element in $Y(v)$ we store  information about the chunk of its corresponding position in $Y(u)$  using  the binary sequence $C(v)$: $C(v)$ contains a $1$ for every element $Y(v)[i]$ and a $0$ for every chunk in $Y(u)$ ($0$ indicates the end of a chunk). We store in  $\up(v)[i]$ the relative value of its corresponding position in $Y(u)$. That is, if the element of $Y(u)$ that corresponds to $Y(v)[i]$ is in the $j$th chunk of $Y(u)$, then it is at $Y(u)[j\cdot 2^{d_k}+\up(v)[i]]$. In order to move from $Y(v)[i]$ in a node $v$ to the corresponding position $Y(u)[i_k]$ in its $d_k$-ancestor $u$, we compute the target chunk in $Y(u)$,
 $j=\mathrm{select}_1(C(v),i)-i$, and set $i_k=j\cdot 2^{d_k}+\up(v)[i]$. Here $\mathrm{select}_1$ finds the $i$th $1$ in $C(v)$, and can be computed in constant time using $o(|C(v)|)$ bits on top of $C(v)$ \cite{Cla96,Mun96}. 

Since the tree contains $h/d_{k-1}$ levels of $t$ $d_{k-1}$-nodes, and the 
$UP(v)$ sequences of $d_{k-1}$-nodes $v$ store numbers up to $2^{d_k}$, the 
total space used by all $UP(v)$ sequences for all $d_{k-1}$-nodes $v$ is 
$O(t \cdot (h/d_{k-1})\cdot d_k)=O(t\cdot h^{1+\eps})$ bits, because 
$d_k/d_{k-1}=h^{\eps}$. 
For any such node $v$, with $d_k$-ancestor $u$, the total number of bits 
in $C(v)$ is $|Y(v)|+|Y(u)|/2^{d_k}$. There are at most $2^{d_k}$ 
nodes $v$ with the same $d_k$-ancestor $u$. Hence, summing over all 
$d_{k-1}$-nodes $v$, all $C(v)$s use $t(h/d_{k-1})+ t(h/d_k)=O(t(h/d_{k-1}))$ 
bits. These structures are stored for all values 
$k-1 \in \{0,\ldots,\lceil 1/\eps\rceil-1\}$.
Summing up, all sequences $C(v)$ use $O(t \cdot h)$ bits.
The total space needed by auxiliary structures is then 
$O(t\cdot h^{1+\eps}) = O(n\log^{1+\eps/2}\sigma \log^{\eps/2} n)$ bits, 
dominated by the sequences $UP(v)$. This can be written as
$O(n\log\sigma \log^\eps n)$ bits.

To produce the auxiliary structures, we need essentially that each $d_k$-node
$u$ distributes its positions in the corresponding $C(v)$ and $UP(v)$ structures
in each of the next $h^\eps-1$ levels of $d_{k-1}$-nodes below $u$. Precisely, 
there are
$2^{l \cdot d_{k-1}}$ $d_{k-1}$-nodes $v$ at distance $l \cdot d_{k-1}$ from 
$u$, and we use $l \cdot d_{k-1}$ bits from the coordinates in $Y(u)[i]$ to 
choose the appropriate node $v$ where $Y(u)[i]$ belongs. Doing this in sublinear
time, however, requires some care.

Let us first consider the root $u$, the only $d_k$-node for $k=\ceil{1/\eps}$. We consider all the $d_{k-1}$-nodes $v$ (thus, $u$ is their only $d_k$-ancestor). These are nodes of height $l\cdot d_{k-1}$ for $l=1, 2, \ldots, h^{\eps}-1$. In order to construct sequences $\up(v)$ in all nodes $v$ on level $l\cdot d_{k-1}$ for a fixed $l$, we proceed as follows. The sequence $Y[u]$ is divided into chunks, so that each chunk contains $2^h$ consecutive elements. The elements $Y(u)[i]$ within each chunk are sorted with key pairs $(\bits((h^{\eps}-l)\cdot d_{k-1}, Y(u)[i]),\pos(i,u))$ where $\pos(i,u) = i \!\!\mod 2^h$ is the relative position of $Y(u)[i]$ in its chunk and $\bits(\ell,x)$ is the  number that consists of the highest $\ell$ bits of $x$. We sort integer pairs in the chunk using a modification of the algorithm of Albers and Hagerup~\cite[Thm.~1]{AlbersH97} that runs in $O(2^h\frac{h^2}{\log n})$ time. Our modified algorithm  works in the same way as the second phase of their algorithm, but we merge words in $O(1)$ time. Merging can be implemented using a universal look-up table that uses $O(\sqrt{n})$ words of space and can be initialized in $O(\sqrt{n}\log^3 n)$ time. 

We then traverse the chunks and generate the sequences $\up(v)$ and $C(v)$ for all the nodes $v$ on level $l\cdot d_{k-1}$. For each bit string of length $l\cdot d_{k-1}$, we say that $v$ is the $q$-descendant of $u$ if the path from $u$ to $v$ is labeled with $q$. The sorted list of pairs for each chunk of $u$ is processed as follows. All the pairs $(q,\pos(i,u))$ (i.e., $q=\bits((h^{\eps}-l)d_{k-1},Y(u)[i])$) are consecutive after sorting, so we scan the list identifying the group for each value of $q$; let $n(q)$ be its number of pairs. Precisely, the points with value $q$ must be stored at the $q$-descendant $v$ of $u$ (the consecutive values of $q$ correspond, left-to-right, to the nodes $v$ on level $l\cdot d_{k-1}$). For each group $q$, then, we identify the $q$-descendant $v$ of $u$ and append 
$n(q)$ $1$-bits and one $0$-bit to $C(v)$. We also append $n(q)$ entries to $\up(v)$ with the contents $\pos(i,u)$, in the same order as they appear in the chunk of $u$. 

We need time $O(2^h \cdot h/\log n)$ to generate the pairs $(\bits(\cdot),\pos(\cdot))$ for the $2^h$ coordinates of each chunk, and to store the pairs in compact form, that is, $O(\log(n)/h)$ pairs per word. We can then sort the chunks in time
$O(2^h \cdot h^2/\log n)=O(2^h\log\sigma)$.  
We can generate the parts of sequences $C(v)$ and $\up(v)$ that correspond to a chunk for all nodes $v$ on level $l\cdot d_{k-1}$  in $O(2^h+ 2^h \cdot h/\log n)=O(2^h)$. 
Thus the total time needed to generate $\up(v)$ and $C(v)$ for all nodes $v$ on level $l\cdot d_{k-1}$ and some fixed $l$ is $O(t\log\sigma)$, where we remind that $t$ is the total number of elements in the root node.  The total time needed to construct $\up(v)$ and $C(v)$ for all $d_{k-1}$-nodes $v$ is then $O(th^\eps\log \sigma)$.


Now let $u$ be an arbitrary $d_k$-node. Using almost the same method as above, we can produce sequences $\up(v)$ and $C(v)$ for all $(d_{k-1})$-nodes $v$, such that $u$ is a $d_k$-ancestor of $v$. There are only two differences with the method above. First, we divide the sequence $Y(u)$ into chunks of size $2^{d_k}$. Second, the sorting of elements in a chunk is not based on the highest bits, but on a less significant chunk of bits: the pairs are now $(\bitval(Y(u)[i]),\pos(i,u))$. If the bit representation of $Y(u)[i]$ is $b_1b_2\ldots b_d$, then  $\bitval(Y(u)[i])$ is the integer with bit representation $b_{f+1}b_{f+2}\ldots b_{f+d_k}$ where $f$ is the depth of the node $u$ in the range tree. The total time needed to produce $C(v)$ and $\up(v)$ is $O(|Y(u)|d_k/\log n+ |Y(u)| d_k^2/\log n)$, the first term to create the pairs and the second to sort the chunks and produce $C(v)$ and $\up(v)$. The number of different elements in all $d_k$-nodes is $O(t\cdot h/d_k)$, and each produces the sequences of $h^\eps$ levels of $d_{k-1}$-nodes. Hence the time needed to produce the sequences for all $d_{k-1}$-nodes is $O((t \cdot h)/d_k \cdot h^\eps \cdot d_k^2/\log n) = O(t \cdot h^{1+\eps} \cdot d_k/\log n)=O(t(h^2/\log n)h^\eps)=O(t h^\eps \log\sigma) = O(t\log^{1+\eps/2}\sigma\log^{\eps/2}n)$. The complexity stays the same after adding up the $1/\eps$ values of $k$: $O(t\log^{1+\eps/2}\sigma\log^{\eps/2}n) =
O(n \log^{3/2+\eps/2} \sigma / \log^{1/2-\eps/2} n)$, which can be written as
$O(n \log^{3/2}\sigma / \log^{1/2-\eps} n)$.

The data structure supporting select queries on $C(v)$ can be built in 
$O(|C(v)|/\log n)$ time \cite[Thm.~5]{MNV16}. This amounts to $O(th/\log n)
= o(t)$ further time.

\subsection{Predecessors and successors of $x$-coordinates}
\label{sec:predsuc}

Now we describe how predecessor and successor queries on $x$-coordinates of points in $Y(u)$ can be answered for any node $u$ in time $O(\log\log n)$. 

We divide the sequence $Y(u)$ into blocks, so that each block contains $\log n$ points. We keep the minimum $x$-coordinate from every block in a predecessor data structure $Y^b(u)$. In order to find the predecessor of $x$ in $Y(u)$, we first find its predecessor $x''$ in $Y^b(u)$; then we search the block of $x''$ for the predecessor of $x$ in $Y(u)$.

The predecessor data structure finds $x''$ in $O(\log\log n)$ time. We can compute the $x$-coordinate of any point in $Y(u)$ in $O(1)$ time as shown above. Hence the predecessor of $x$ in a block can be found in $O(\log\log n)$ time too, using binary search. We can find the successor analogously.

The sampled predecessor/successor data structures store
$O((n/r)(r\log\sigma)/\log n) = O(n/\log_\sigma n)$ elements over all the
levels. An appropriate construction \cite[Thm.~4.1]{FPS16} builds them in 
linear time ($O(n/\log_\sigma n)$) and space ($O(n\log\sigma)$ bits), once 
they are sorted.

\arxiv{
\subsection{Range Reporting in External Memory}
\label{sec:rangerepext}
 Our method of  constructing an external-memory range reporting  structure follows the same basic range tree approach. However the data structure is much simpler because we do not aim for almost-linear space and linear construction time. 
\begin{lemma}
  \label{lemma:extrep}
 For a set of $t=O(n/r)$ points on $t\times \sigma^{O(r)}$  grid, there is an $O(t\cdot r\cdot \log \sigma)$-word data structure that can be constructed in $O((t/B)\cdot r\cdot \log\sigma)$ I/Os and supports orthogonal  range reporting queries in  $O(\log\log t +\pocc/B)$ I/Os  where $\pocc$ is the number of reported points.
\end{lemma}
\begin{proof}
  We generate the sequences $Y(u)$ in the same way as in the internal-memory data structure.
But we also keep in every node the sequence $X(u)$ that explicitly contains the $x$-coordinates of points. All sequences can be generated in $O((m/B)\log\sigma\log t)$ I/Os. For all points stored in a node, we create a data structure for three-sided queries. 
A data structure for three-sided queries on $m$ points can be constructed in $O(m/B)$ I/Os. 
\end{proof}
}

\arxiv{

\section{Index for Small Patterns in External Memory}

\label{sec:smallext}
In this section we describe the index for small patterns, i.e., for strings of length less than $4r+3$. An occurrence of such a string does not necessarily cross a sampled position in the text, because the sampling step is bigger here. 
Hence our method for internal memory does not work.  

We set $d=9r$ and $d_1=2d$. The text $T$ is regarded as an array $A_0[0..n/d]$; every slot of $A_0$ contains a length-$d_1$ string.  
For every slot $A_0[i]=\alpha$, we record its  \emph{position} $i$ and its value $\alpha$. 
First we sort the elements of $A_0$ by their values. Since $A_0$ consists of $n/d$ words, it can be sorted in $O((n/d)\log_{M/B}n)$ I/Os. Then we construct arrays $A_j[]$ for $1\le j\le d-1$.  $A_j$ is obtained from $A_{j-1}$ 
by removing the first symbols from strings stored in slots of $A_{j-1}$ and sorting the resulting strings. Thus each slot of $A_j[t]$ contains a string of length $d-j$  that starts at position $id+j$ in $T$ for some $0\le i < n/d$; each $A_j[t]$ corresponds to some slot $A_0[t']$ with the first $j$ characters removed. 

We can obtain $A_1$ from $A_0$ in $O(\frac{n}{Bd})$ I/Os. We traverse the array $A_0$ and identify parts of $A_0$ that start with the same symbol. We find indices $k_0$, $k_1$, $\ldots$, $k_{\sigma-1}$, $k_{\sigma}=n$ such that all $A_0$ start with the same character: for all $i$, $k_j\le i<k_{j+1}$, $A_0[i]=b_j\alpha_i'$ 
for the $j$-th alphabet symbol $b_j$ and some length-$(d_1-1)$ string $\alpha'_i$. Next, we merge sub-arrays $A_0[k_i..k_{i+1}-1]$ for $i=0$, $1$, $\ldots$, $\sigma-1$. We can merge two-sub-arrays $A[f_1..l_1]$ and $A[f_2..l_2]$ in $O(\frac{l_1+l_2-f_1-f_2}{B}+1)$ I/Os and obtain the new sub-array sorted by $\alpha'_i$. Hence we can obtain the array $A_1$ with values sorted by $\alpha'_i$ in $O(\frac{n}{B}\log\sigma)$ I/Os. 
Finally we traverse $A_1$ and remove the first symbol from every element.  We can obtain $A_{j+1}$ from $A_j$ in the same way. The cost of constructing $A_{j+1}$ is bounded by $O((n_j/B)\log\sigma)$ where $n_j$ is the size of $A_j$. 
For every array $A_j$,  a dictionary $D_j$ contains all distinct values that occur in $A_j$. $D_j$ is implemented as a van Emde Boas data structure  and can be constructed in $O((m_j/B)+(n/Bd))$ I/Os, where $m_j$ is the number of elements in $D_j$. We observe that $m_j$ is bounded by the number of distinct strings with length $(d-j)$, $m_j\le (1/\sigma)^j\cdot \sigma^d$. Since $d=O(r)$ and $r\le \log_{M/B} n$, all $A_j$ and $D_j$ are constructed in $O((n/Bd)\log_{M/B}n)+ d\cdot ((n/Bd)\log\sigma)+ (\sigma^d)/B)=O((n/B)\log\sigma+ (1/r)\sort(n))$ I/Os. 

Suppose that we want to report all occurrences of a string $Q$. Let $Q_s$ be the lexicographically smallest length-$d_1$ string that starts with $Q$ and let $Q_m$ be the lexicographically largest length-$d_1$ string that starts with $Q$.  Strings $Q_s$ and $Q_m$ correspond  to  integers $q_s$ and $q_m$ respectively. We query data structures $D_0$, $D_1$, $\ldots$, $D_{d-1}$. For every $j$, $0\le j \le d-1$, we find $\newsucc(q_s,D_j)$ where $\newsucc(q_s,D_j)$ is the smallest element $y\in D_j$ such that $y\ge q_s$. We traverse the sorted list of elements in $A_j$ until an element $y'> q_m$ is found. For every element between $y$ and $y'$, we report its position in $T$.  We answer $r$ successor queries in $O(r\log\log n)$ I/Os. A pattern matching  query is answered in $O(r\log\log n+ \occ/B)$ I/Os. 
\begin{lemma}
  \label{lemma:smallextmem}
There exists a data structure that uses $O(n/B)$ blocks of space and reports all occurrences of a query string $Q$ in $O(r\log\log n+\occ/B)$ I/Os if $|Q|\le 4r+3$. This data structure can be constructed in $O((n/B)\log\sigma + (1/r)\sort(n))$ I/Os. 
\end{lemma}

\section{Marked Ancestor Queries}
\label{sec:markanc}
We describe a data structure for the nearest marked ancestor queries used in the proof of Lemma~\ref{lemma:jumpy0}. Nodes in a tree $\cT$ are marked with strings of length at most $g$ for $g=4r+3$. A node can be marked with many strings, but the total number of strings that mark the nodes of $\cT$ does not exceed $m'=m\log_{\sigma} n$ where $m$ is the number of nodes. We assume that every internal node has at least two children.  
\begin{lemma}
  \label{lemma:marked}
Suppose that the height of $\cT$ is bounded by $O(B\log^3 n)$. 
There exists an $O(m')$-space data structure for the marked ancestor queries that answers nearest marked ancestor queries in $O(1)$ I/Os. 
\end{lemma}
\begin{proof}
  We consider a heavy-path decomposition of $\cT$ and represent it as a tree $\bT$. Each  node $\nu$ of $\bT$ represents some heavy path $\pi_{\nu}$ in $\cT$. The node $\nu\in \bT$ is a child of a node  $\mu\in \bT$ if the highest node on $\nu$ is a child of some node on $\mu$.  Let $hdepth(u)$ denote the depth of the heavy path that contains  a node $u$ in $\bT$. That is, $hdepth(u)=h_u-1$ where $h_u$ is the number of heavy paths traversed when we ascend from $u$ to the root node.   For a node $u$ let $left(u)$ and $right(u)$ denote  indices of its leftmost and  rightmost leaf descendants respectively.

For every string $s$, we represent nodes marked with $s$ as three-dimensional points and store them in a  data structure $D_s$. If a node $u$ on a heavy path $\mu$ is marked with $s$, then $D_s$ contains a point $(left(u),right(u),hdepth(\mu))$. $D_{s}$ answers the following queries: for any $a<b$ report some point $(l,r,h)\in D_s$ such that $l\le a$, $r\ge b$, and $h$ is maximized.  By extending the result described in~\cite[Lemma 8]{ChanNRT18}, we can implement $D_s$ in linear space, so that queries are answered in $O(1)$ time. 
Details will be given in the full version of the paper. 

For every heavy path $\mu$ that contains at least one node marked with a string $s$, we keep a data structure $E_{\mu,s}$. $E_{\mu,s}$ contains the depth of all nodes in $\mu$ marked with $s$ and answers predecessor queries. Since the depth of $T$ is bounded by $B^3\log^3 n$, each data structure $E_{\mu,s}$ stores $O(B^3 \log^3 n)$ elements. Hence, $E_{\mu,s}$ supports predecessor queries  in $O(1)$ I/Os: if $B> \log^3 n$, we implement $E_{\mu,s}$ as a B-tree and if $B\le\log^3 n$, we implement $E_{\mu,s}$ as a fusion tree.  Each node  $u$ is represented  by a point $(left(u), right(u),hdepth(u))$. 

In order to find the lowest ancestor $u$  of a node $v$ marked with a string $s$, we identify the heavy path $\mu$ that contains  $u$.  This can be done in $O(\log\log n)$ I/Os using $D_s$: if some heavy path $\nu$ contains an ancestor $u$ of $v$ marked with $s$, then $D_s$ contains a point $(left(u),right(u),hdepth(\nu))$ satisfying $left(u)\le left(v)$, $right(u)\ge right(v)$. Suppose that $u$ is the lowest marked ancestor of $v$ and $u$ is on a heavy path $\mu$. Then $D_s$ contains at least one point $(l,r,h)$ satisfying $l\le left(v)$, $r\ge right(v)$, and $h=hdepth(\mu)$. Also there is no point $(l',r',h')$ in $D_s$ such that  $l'\le left(v)$, $r'\ge right(v)$, and $h'>hdepth(\mu)$. 

We find a point $(left(u'),right(u'),hdepth(u'))$ in $D_s$ such that $left(u')\le left(v)$, $right(u')\ge right(v)$ and $hdepth(u')$ is minimized.  Clearly $u'$ is on the same heavy path $\mu$ as the lowest marked ancestor. When $\mu$ is found, we can compute the lowest marked ancestor $u$ using $E_{\mu,s}$; it suffices to find the predecessor of $depth(v)$ in $E_{\mu,s}$ where $depth(v)$ denotes the depth of the node $v$.  Since $\mu$ is found in $O(\log\log n)$ I/Os and $u$ is found in $O(1)$ I/Os, the total query cost is $O(\log\log n)$. 
\end{proof}

\begin{lemma}
  \label{lemma:domin-extend}
Let $V$ denote a set of $n$ three-dimensional points $(x,y,z)$ satisfying $1\le x,y\le n$ and $1\le z\le\log n$. There exists an $O(n)$-space data structure that answers the following queries in $o(1)$ time:  for any $q_x$ and $q_y$, find the   highest value $z$ such that $(x,y,z)\in V$,  $x\le q_x$,  and $y\ge q_y$.
\end{lemma}

\section{Weighted Level Ancestor Queries}
\label{sec:wlacons}
A weighted level ancestor query $(u,x)$ on a suffix tree $\cT$ asks for the highest ancestor $w$ of the node $u\in \cT$ such that the string depth of $w$ is not smaller than $x$.
\begin{lemma}
  \label{lemma:wlacons}
There exists a data structure that answers weighted level ancestor queries on a suffix tree $\cT$ in $O(\log \log m)$ I/Os, where $m$ is the total number of nodes in $\cT$. This data structure uses space $O(m)$ and can be constructed in $O(\sort(m))$ I/Os.
\end{lemma}
\begin{proof}
Let $\cT$ denote a  suffix tree with $m$ nodes. We will say that a node $u$ of $\cT$ is \emph{big} if it has at least $\max(B/2,\log n)$ leaf descendants, otherwise a node is  small.  We visit all nodes of a tree and identify all big nodes $u$. We also identify all subtrees $\cT_v$ of $\cT$, such that the root $v$  of $\cT_v$ is a small node, but the parent of $v$ is a big node. If $B> 2\log n$, we store all subtrees $\cT_v$ in such way  that  each $\cT_v$ fits into at most two blocks of memory. Traversing the tree in depth-first order~\cite{ChiangGGTVV95}, we can classify all nodes into big and small and identify all subtrees $\cT_v$ in $O(\sort(m))$ I/Os, where $m$ is the total number of nodes.

 Let $\cT^b$ denote the subtree of $\cT$ induced by all big nodes. The total number of leaves in $\cT^b$ is $O(n/B)$. We construct the weighted level ancestor data structure on $\cT^b$. We compute the centroid path decomposition of $\cT^b$.  Weighted level ancestor queries on $\cT^b$ can be answered as follows. For every centroid path $\pi$, we construct a predecessor data structure. We also store the string depth of the highest node on every heavy path.  Furthermore for every leaf-to-root path  $\pi_{v}$ in $\cT^b$ we construct a separate data structure; this data structure  stores the string depth of the highest node on every centroid path that is intersected by $\pi_{v}$ and answers predecessor queries.   We can generate the centroid path decomposition of $\cT^b$ in $O(m/B)$ I/Os. All predecessor data structures on centroid paths can be constructed in $O(m/B)$ I/Os because the total number of nodes in $\cT^b$ is $O(m)$ and the number of paths is $O(m/B)$.  We can generate predecessor data structures on all paths $\pi_v$ in $O(\sort(m))$ I/Os using Euler tour of $\cT^b$~\cite{ChiangGGTVV95}. Finally we create data structures for leaf-to-root paths $\pi_v$. We execute the Euler tour of $\cT^b$ and mark all nodes $u$, such that $u$ is the  highest node in its centroid path. Then we traverse the list of nodes (in the Euler-tour order) one more time. Every time when a marked node $u$ is visited for the first time, we add the node $u$ and its string depth  to the list $C$. When 
a leaf node $v$ is visited, we copy all nodes from $C$ to the data structure for the root-to-leaf path $\pi_v$. When a marked node $u$ is visited for the second time, we remove $u$  from the list $C$. If $B> \log n$, the  cost of traversing the tree $\cT^b$ and updating $C$ is $O(m/B)$. The cost of generating data structures $\pi_v$ is $O(\frac{m\log n}{B^2})=O(m/B)$ because $\cT^b$ has $O(m/B)$ leaves. If $B< \log n$, the cost of traversing $\cT^b$ and updating $C$ is also $O(m/B)$ and the cost of generating $\pi_v$ is $O(\frac{(m/B)\log n}{\log n})=O(m/B)$ because $\cT^b$ has $O(n/\log n)$ leaves.

Now we describe weighted ancestor structures for subtrees $\cT_v$. If $B> 2\log n$,  each$\cT_v$ is stored in at most two blocks of memory; hence a weighted level  ancestor  query on $\cT_v$ is answered in $O(1)$ I/Os. If $B< 2\log n$, each $\cT_v$ has $O(\log n)$ nodes. In this case our approach is very similar to the data structure used for $\cT^b$. We can obtain the centroid path decomposition of $\cT_v$ in $O(\sort(m_v))$ I/Os where $m_v$ is the number of nodes in $\cT_v$. With  each node of $\cT_v$ we store its centroid path $\mu$ and  the highest node on the centroid path $\mu$. We also keep a pointer to the parent node and the string depth of each node. For every centroid path we store a predecessor data structure described in  the previous paragraph.

A weighted ancestor query for a leaf $\ell$ and a value $x$ can be answered as follows. We identify the subtree $\cT_v$ that contains $\ell$. Let $v$ be the root node of the subtree $\cT_v$. If the string depth of $v$ is larger  than $x$, we answer a query on $\cT^b$. For the path $\pi_v$ from $v$ to the root node, we find the highest centroid path $\pi_c$ intersected by $\pi_v$ such that the string depth of the highest node on $\pi_c$ is larger than $x$. 
Let $\pi'_c$ denote the centroid path directly above $\pi_c$ (i.e., the path $\pi'_c$  contains the parent of the highest node on $\pi_c$). Using the predecessor data structure, we find the highest node on $\pi'_c$ with string depth larger than $x$. This node is the weighted level ancestor. 

If the string depth of $v$ is smaller  than $x$, we answer the ancestor query on $\cT_v$. If $B> 2\log n$, $\cT_v$ fits into $O(1)$ blocks and the query is answered in $O(1)$ I/Os. 
If $B< \log n$, $\cT_v$ has $O(\log n)$ nodes. Let $\pi_{\ell}$ denote the path from $\ell$ to $v$. We visit all centroid paths that intersect $\pi_{\ell}$ (starting with the lowermost path) until we find the centroid path $\pi_c$, such that the string depth of the highest node on $\pi_c$ is smaller than $x$. Since $\cT_v$ has $O(\log n)$ nodes, $\pi_{\ell}$ intersects $O(\log\log n)$ centroid paths. Hence we can find $\pi_c$ in $O(\log \log n)$ I/Os. Using the predecessor data structure on $\pi_c$, we can find the weighted level ancestor in $O(\log \log n)$ I/Os. 
 
\end{proof}

Using the same approach, we can also construct a data structure for the marked ancestor queries, described in Lemma~\ref{lemma:marked}. 

}

\bibliographystyle{plain}
\bibliography{paper}

\end{document}